\documentclass[12pt,draftcls,onecolumn]{IEEEtran}

\usepackage{amsmath,graphicx,amsfonts,amssymb,epsfig,subfigure,mathrsfs}
\usepackage{cite}
\usepackage{color}
\usepackage{algorithm}
\usepackage{algpseudocode}

\ifCLASSINFOpdf
\else
\fi
\usepackage[pdfa]{hyperref}

\newtheorem{theorem}{Theorem}

\newtheorem{lemma}{Lemma}

\newtheorem{remark}{Remark}

\begin{document}

\title{A Probabilistic Method for Nonlinear Robustness Analysis of F-16 Controllers}

\author{Abhishek~Halder$^{*}$,
        Kooktae~Lee,
        and~Raktim~Bhattacharya
\thanks{$*$ Corresponding author. This research was supported by NSF grant CSR-1016299 with Dr. Helen Gill as the program manager. Some preliminary results were presented \cite{HalderLeeBhattacharya2013} in American Control Conference, 2013.}
\thanks{The authors are with the Department
of Aerospace Engineering, Texas A\&M University, College Station,
TX 77843-3141, USA. e-mail: \texttt{\{ahalder,animodor,raktim\}@tamu.edu}}
}

\maketitle

\begin{abstract}
This paper presents a new framework for controller robustness verification with respect to F-16 aircraft's closed-loop performance in longitudinal flight. We compare the state regulation performance of a linear quadratic regulator (LQR) and a gain-scheduled linear quadratic regulator (gsLQR), applied to nonlinear open-loop dynamics of F-16, in presence of stochastic initial condition and parametric uncertainties, as well as actuator disturbance. We show that, in presence of initial condition uncertainties alone, both LQR and gsLQR have comparable immediate and asymptotic performances, but the gsLQR exhibits better transient performance at intermediate times. This remains true in the presence of additional actuator disturbance. Also, gsLQR is shown to be more robust than LQR, against parametric uncertainties. The probabilistic framework proposed here, leverages transfer operator based density computation in exact arithmetic and introduces optimal transport theoretic performance validation and verification (V\&V) for nonlinear dynamical systems. Numerical results from our proposed method, are in unison with Monte Carlo simulations.
\end{abstract}


\section{Introduction}

\IEEEPARstart{I}{n} recent times, the notion of probabilistic robustness \cite{barmish1996probabilistic,calafiore2000randomized,polyak2001probabilistic,wang2002robust,fujisaki2003probabilistic,tempo2005randomized,
chen2005risk}, has emerged as an attractive alternative to classical worst-case robust control framework. There are two key driving factors behind this development. \textbf{First}, it is well-known \cite{chen2005risk} that the deterministic modeling of uncertainty in the worst-case framework leads to conservative performance guarantees. In particular, from a probabilistic viewpoint, classical robustness margins can be expanded significantly while keeping the risk level acceptably small \cite{barmish1997,lagoa1997probabilistic,lagoa2003probabilistic}. \textbf{Second}, the classical robustness formulation often leads to problems with enormous computational complexity \cite{khargonekar1996randomized,tempo1996probabilistic,vidyasagar2001probabilistic}, and in practice, relies on relaxation techniques for solution.

Probabilistic robustness formulation offers a promising alternative to address these challenges. Instead of the interval-valued structured uncertainty descriptions, it adopts a risk-aware perspective to analyze robustness, and hence, explicitly accounts the distributional information associated with unstructured uncertainty. Furthermore, significant progress have been made in the design and analysis of randomized algorithms \cite{tempo2005randomized,calafiore2011research} for computations related to probabilistic robustness. These recent developments are providing impetus to a transition from ``worst-case" to ``distributional robustness" \cite{lagoa2002distributionally,nagy2003worst}.

\subsection{Computational challenges in distributional robustness}

In order to fully leverage the potential of distributional robustness, the associated computation must be scalable and of high accuracy. However, numerical implementation of most probabilistic methods rely on Monte Carlo like realization-based algorithms, leading to high computational cost for implementing them to nonlinear systems. In particular, the accuracy of robustness computation depends on the numerical accuracy of histogram-based (piecewise constant) approximation of the probability density function (PDF) that evolves spatio-temporally over the joint state and parameter space, under the action of closed-loop nonlinear dynamics. Nonlinearities at trajectory level cause non-Gaussianity at PDF level, even when the initial uncertainty is Gaussian. Thus, in Monte Carlo approach, at any given time, a high-dimensional nonlinear system requires a dense grid to sufficiently resolve the non-Gaussian PDF, incurring the `curse of dimensionality' \cite{Bellman1957}.

This is a serious bottleneck in applications like flight control software certification \cite{chakraborty2011}, where the closed loop dynamics is nonlinear, and linear robustness analysis supported with Monte Carlo, remains the state-of-the art. Lack of nonlinear robustness analysis tools, coupled with the increasing complexity of flight control algorithms, have caused loss of several F/A-18 aircrafts due to nonlinear ``falling leaf mode" \cite{seiler2012assessment}, that went undetectable \cite{chakraborty2011susceptibility} by linear robustness analysis algorithms. On the other hand, accuracy of sum-of-squares optimization-based \emph{deterministic} nonlinear robustness analysis \cite{chakraborty2011,seiler2012assessment} depends on the quality of semi-algebraic approximation, and is still computationally expensive for large-scale nonlinear systems. Thus, there is a need for controller robustness verification methods, that does not make any structural assumption on nonlinearity, and allows scalable computation while accommodating stochastic uncertainty.

\subsection{Contributions of this paper}

\subsubsection{PDF computation in exact arithmetic}
Building on our earlier work \cite{HalderBhattacharya2010GNC,HalderBhattacharya2011JGCD}, we show that stochastic initial condition and parametric uncertainties can be propagated through the closed-loop nonlinear dynamics \emph{in exact arithmetic}.
This is achieved by leveraging the fact that the transfer operator governing the evolution of joint densities, is an infinite-dimensional \emph{linear} operator, even though the underlying finite-dimensional closed-loop dynamics is \emph{nonlinear}. Hence, we directly solve the linear transfer operator equation subject to the nonlinear dynamics.

This crucial step distinguishes the present work from other methods for probabilistic robustness computation by explicitly using the exact values of the joint PDF instead of empirical estimates of it. Thus, from a statistical perspective, the robustness verification method proposed in this paper, is an \emph{ensemble formulation} as opposed to the \emph{sample formulations} available in the literature \cite{tempo1996probabilistic,tempo2005randomized}.

\subsubsection{Probabilistic robustness as optimal transport distance on information space}
Based on Monge-Kantorovich optimal transport \cite{Villani2003,Villani2008}, we propose a novel framework for computing probabilistic robustness as the ``distance" on information space. In this formulation, we measure robustness as the minimum effort required to transport the probability mass from instantaneous joint state PDF to a reference state PDF. For comparing regulation performance of controllers with stochastic initial conditions, the reference state PDF is Dirac distribution at trim. If in addition, parametric uncertainties are present, then the optimal transport takes place on the extended state space with the reference PDF being a degenerate distribution at trim value of states. We show that the optimal transport computation is meshless, non-parametric and computationally efficient. We demonstrate that the proposed framework provides an intuitive understanding of probabilistic robustness while performing exact ensemble level computation.

\subsection{Structure of this paper}
Rest of this paper is structured as follows. In Section II, we describe the nonlinear open-loop dynamics of F-16 aircraft in longitudinal flight. Section III provides the synthesis of linear quadratic regulator (LQR) and gain-scheduled linear quadratic regulator (gsLQR) -- the two controllers whose state regulation performances are being compared. The proposed framework is detailed in Section IV and consists of closed-loop uncertainty propagation and optimal transport to trim. Numerical results illustrating the proposed method, are presented in Section V. Section VI concludes the paper.

\subsection{Notations}
The symbol $\nabla$ stands for the (spatial) gradient operator, and diag(.) denotes a diagonal matrix. Abbreviations ODE and PDE refer to ordinary and partial differential equation, respectively. The notation $\mathcal{U}\left(\cdot\right)$ denotes uniform distribution, and $\delta\left(x\right)$ stands for the Dirac delta distribution. Further, $\text{dim}\left(S\right)$ denotes the dimension of the space in which set $S$ belongs to, and $\text{supp}\left(\cdot\right)$ denotes the support of a function.


\section{F-16 Flight Dynamics}

\subsection{Longitudinal Equations of Motion}
The longitudinal equations of motion for F-16 considered here, follows the model given in \cite{nguyen1979simulator,StevensLewis1992,bhattacharya2001nonlinear}, with the exception that we restrict the maneuver to a constant altitude ($h = 10,000$ ft) flight. Further, the north position state equation is dropped since no other longitudinal states depend on it. This results a reduced \emph{four state, two input model} with $x = \left(\theta, V, \alpha, q\right)^{\top}$, $u = \left(T, \delta_{e}\right)^{\top}$, given by
\small
\begin{subequations}
 \begin{eqnarray}
 \dot{\theta} &=& q, \label{ThetaDot}\\
 \dot{V} &=& \frac{1}{m}\cos\alpha\left[T - mg\sin\theta + \overline{q}S \left(C_{X} + \displaystyle\frac{\overline{c}}{2V} C_{X_{q}} q\right)\right] + \frac{1}{m}\sin\alpha\left[mg\cos\theta + \overline{q}S \left(C_{Z} + \displaystyle\frac{\overline{c}}{2V} C_{Z_{q}} q\right)\right], \label{VDot}\\
 \dot{\alpha} &=& q - \displaystyle\frac{\sin\alpha}{m V} \left[T - mg\sin\theta + \overline{q}S \left(C_{X} + \displaystyle\frac{\overline{c}}{2V} C_{X_{q}} q\right)\right] + \displaystyle\frac{\cos\alpha}{m V} \left[mg\cos\theta + \overline{q}S \left(C_{Z} + \displaystyle\frac{\overline{c}}{2V} C_{Z_{q}} q\right)\right], \label{AlphaDot}\\
 \dot{q} &=& \displaystyle\frac{\overline{q}S\overline{c}}{J_{yy}} \left[C_{m} + \displaystyle\frac{\overline{c}}{2V}C_{m_{q}} q + \displaystyle\frac{\left(x^{\text{ref}}_{\text{cg}} - x_{\text{cg}}\right)}{\overline{c}}\left(C_{Z} + \displaystyle\frac{\overline{c}}{2V}C_{Z_{q}} q\right)\right]. \label{QDot}
 \end{eqnarray}
\label{F16FourStateLongitudinal}
\end{subequations}
\normalsize
\noindent
The state variables are second Euler angle $\theta$ (deg), total velocity $V$ (ft/s), angle-of-attack $\alpha$ (deg), and pitch rate $q$ (deg/s), respectively. The control variables are thrust $T$ (lb), and elevator deflection angle $\delta_{e}$ (deg). Table \ref{ParamF16} lists the parameters involved in (\ref{F16FourStateLongitudinal}). Furthermore, the dynamic pressure $\overline{q} = \frac{1}{2}\rho\left(h\right) V^{2}$, where the atmospheric density $\rho\left(h\right) = \rho_{0}\left(1 - 0.703\times10^{-5}h\right)^{4.14} = 1.8 \times 10^{-3}$ slugs/ft\textsuperscript{3} remains fixed.
\begin{table}
 \begin{center}
\caption{Parameters in eqn. (\ref{F16FourStateLongitudinal})}
\label{ParamF16}
  \begin{tabular}{|l|l|}\hline
\textbf{Description of parameters} & \textbf{Values with dimensions} \\ \hline\hline
Mass of the aircraft & $m = 636.94$ slugs \\
Acceleration due to gravity & $g = 32.17$ ft/s\textsuperscript{2}\\
Wing planform area & $S = 300$ ft\textsuperscript{2}\\
Mean aerodynamic chord & $\overline{c} = 11.32$ ft\\
Reference $x$-position of c.g. & $x_{\text{cg}}^{\text{ref}} = 0.35\:\overline{c}$ ft\\
True $x$-position of c.g. & $x_{\text{cg}} = 0.30\:\overline{c}$ ft\\
Pitch moment-of-inertia & $J_{yy} = 55,814$ slug-ft\textsuperscript{2}\\
Nominal atmospheric density & $\rho_{0} = 2.377\times10^{-3}$ slugs/ft\textsuperscript{3}\\\hline
 \end{tabular}
 \end{center}
\end{table}

\subsection{Aerodynamic Coefficients}

The aerodynamic force and moment coefficients $C_{X}, C_{Z}$, and $C_{m}$ are functions of $\alpha$
and $\delta_{e}$, expressed as look-up table from wind tunnel test data \cite{nguyen1979simulator,StevensLewis1992,bhattacharya2001nonlinear}. Similarly, the stability derivatives
$C_{X_{q}}, C_{Z_{q}}$, and $C_{m_{q}}$ are look-up table functions of $\alpha$. We refer the readers to above references for details.


\section{F-16 Flight Control Laws}
In this paper, we consider two controllers: LQR and gsLQR, as shown in Fig. \ref{ClosedLoopBlockDiagrams}, with the common objective of regulating the state to its trim value. Both controllers minimize the infinite-horizon cost functional
\begin{eqnarray}
\mathcal{J} = \displaystyle\int_{0}^{\infty} \left(x(t)^{\top} Q \:x(t) + u(t)^{\top} R \:u(t)\right) \: dt,
\label{QuadraticCost}
\end{eqnarray}
with $Q = \text{diag}\left(100, 0.25, 100, 10^{-4}\right)$, and $R = \text{diag}\left(10^{-6}, 625\right)$. The control saturation shown in the block diagrams, is modeled as
\begin{eqnarray}
1000\,\text{lb} \leqslant T \leqslant 28,000\,\text{lb}, \quad -25^{\circ} \leqslant \delta_{e} \leqslant + 25^{\circ}.
\label{ControlBounds}
\end{eqnarray}

\begin{figure}[tb]
\hfill
\subfigure[Block diagram for LQR closed-loop system.]{\includegraphics[width=0.49\textwidth]{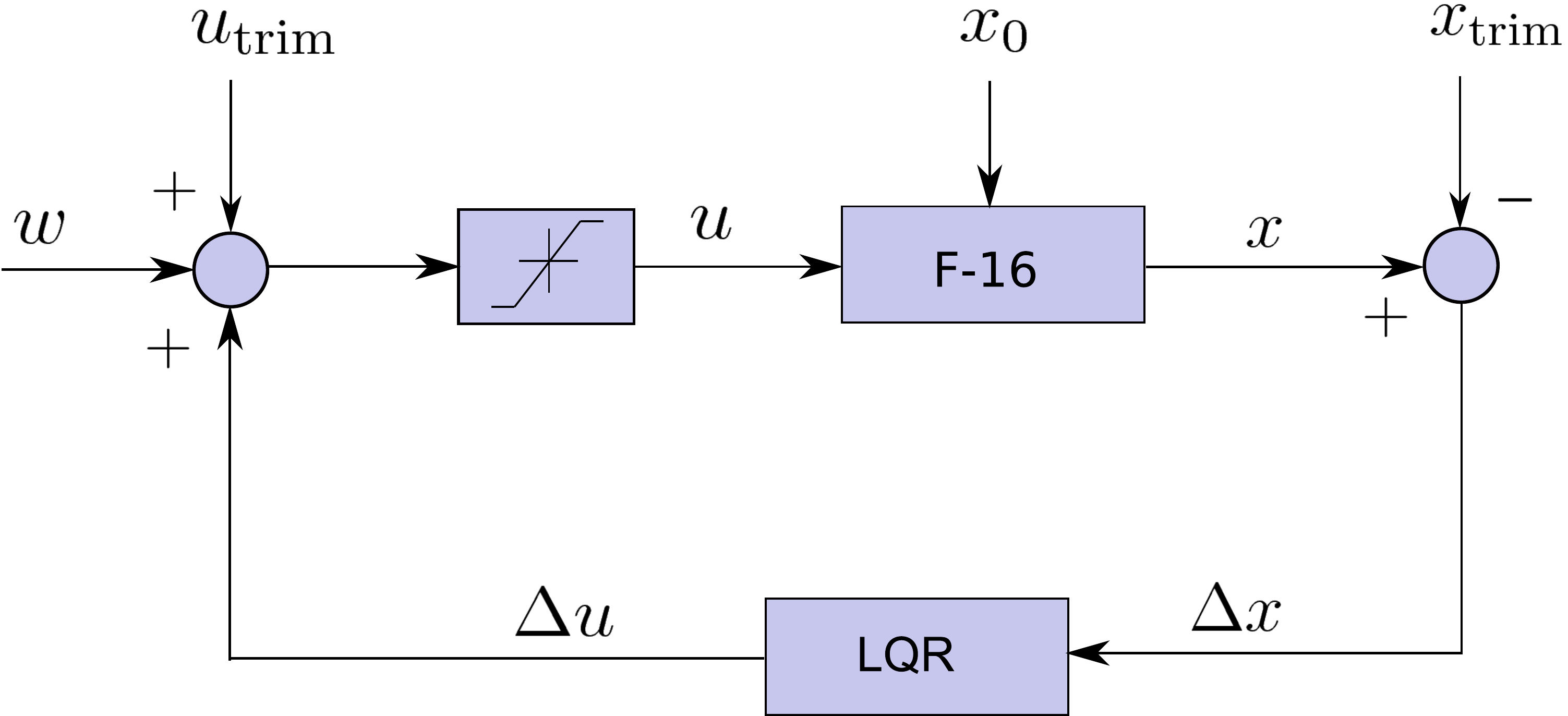}}
\hfill
\subfigure[Block diagram for gsLQR closed-loop system.]{\includegraphics[width=0.49\textwidth]{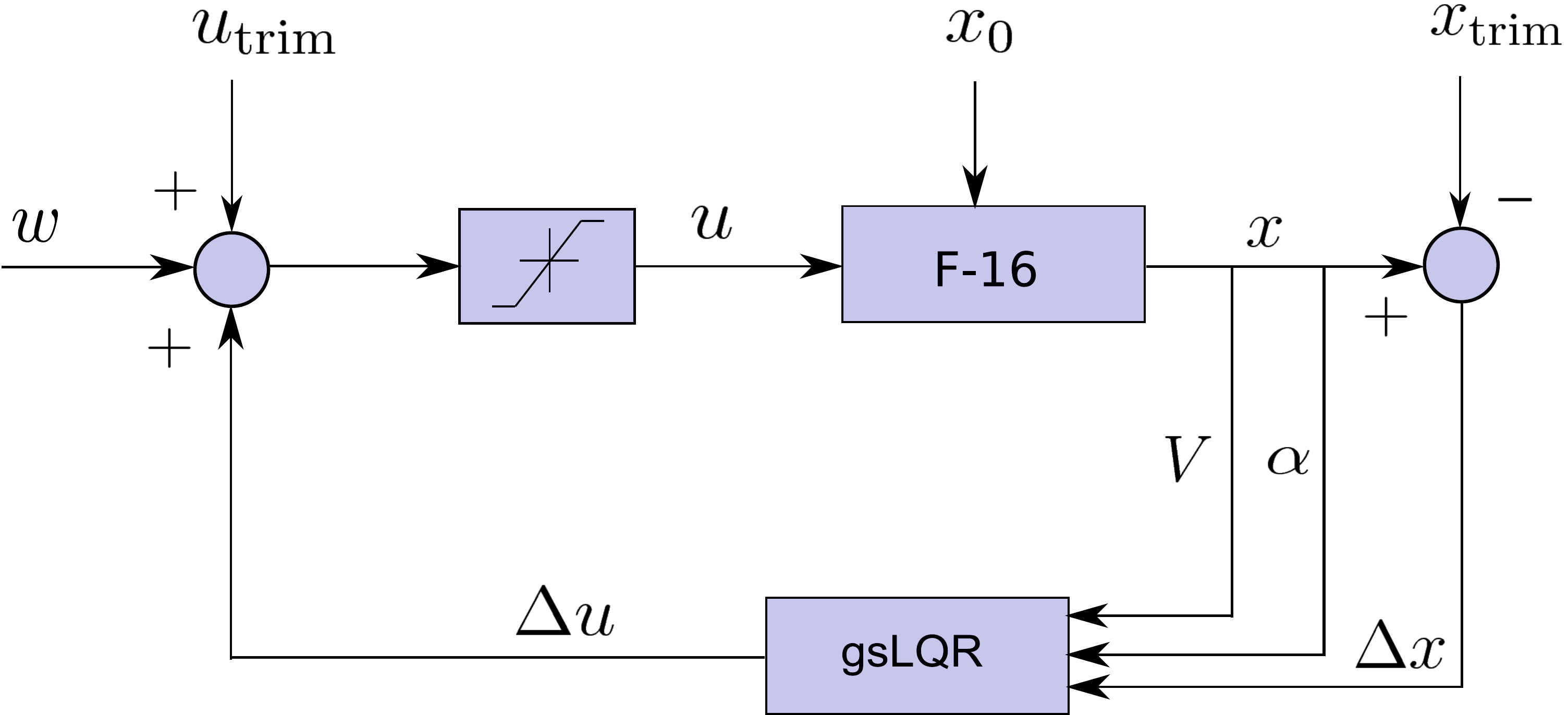}}
\hfill
\caption{Block diagrams for the closed-loop nonlinear systems with (a) LQR and (b) gsLQR controller. Here, $w$ denotes the actuator disturbance.}
\label{ClosedLoopBlockDiagrams}
\end{figure}

\subsection{LQR Synthesis}
The nonlinear open loop plant model was linearized about $x_{\text{trim}}, u_{\text{trim}}$, using simulink \texttt{linmod} command. The trim conditions were computed via the nonlinear optimization package SNOPT \cite{SNOPTversion7}, and are given by $x_{\text{trim}} = \left(2.8190\,\text{deg}, 407.8942 \,\text{ft/s}, 6.1650\,\text{deg}, 6.8463\times10^{-4}\,\text{deg/s}\right)^{\top}$, $u_{\text{trim}} = \left(1000\,\text{lb}, \: -2.9737\,\text{deg}\right)^{\top}$. The LQR gain matrix $K$, computed for this linearized model, was found to be
\begin{eqnarray}
K = \begin{bmatrix} 7144.9 & -400.58 & -1355.8 & 2002.8\\
0.7419 & -0.0113 & -0.2053 & 0.3221
\end{bmatrix}.
\label{LQRgain}
\end{eqnarray}
As observed in Fig. \ref{EigenLQR} (a), both open-loop and LQR closed-loop linear systems are stable.
\begin{figure}
\begin{center}
\includegraphics[scale=0.35]{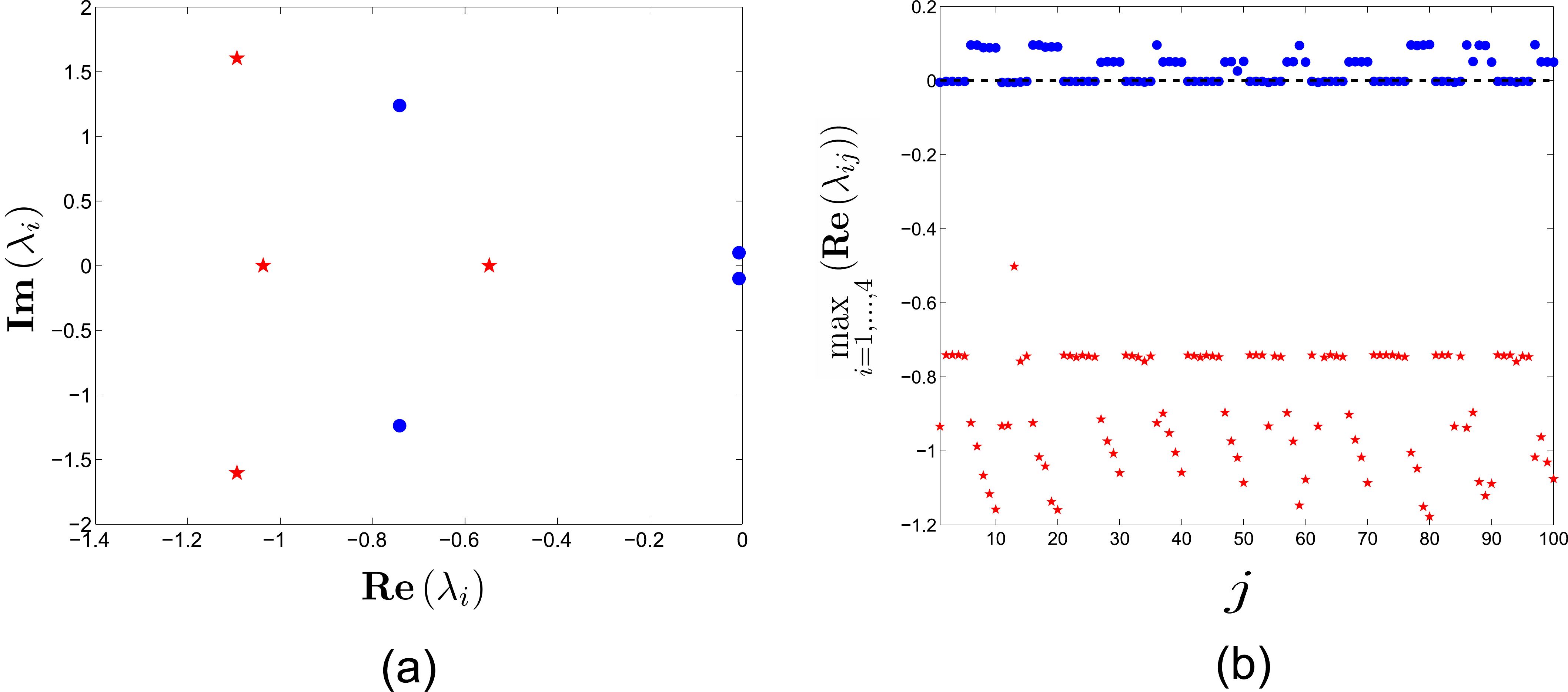}
\centering
\caption{(a) The open-loop (\emph{circles}) and LQR closed-loop (\emph{stars}) eigenvalues shown in the complex plane, for the linearized model. (b) For gsLQR synthesis, maximum of the real parts of open-loop (\emph{circles}) and closed-loop (\emph{stars}) eigenvalues for each of the $j=1,\hdots,100$ linearizations are plotted. Depending on the trim condition, some open-loop linearized plants can be unstable but all closed-loop synthesis are stable.}
\label{EigenLQR}
\end{center}
\end{figure}

\subsection{Gain-scheduled LQR Synthesis}
As shown in Fig. \ref{ClosedLoopBlockDiagrams} (b), $V$ and $\alpha$ are taken as the scheduling states. We generate 100 grid points in the box
\begin{eqnarray}
100 \, \text{ft/s}\: \leqslant V \leqslant 1000\, \text{ft/s}, \quad -10^{\circ} \leqslant \alpha \leqslant +45^{\circ},
\label{SchedulingGrid}
\end{eqnarray}
and compute trim conditions $\{x_{\text{trim}}^{j}, u_{\text{trim}}^{j}\}_{j=1}^{100}$, using SNOPT, for each of these grid points.
Next, we synthesize a sequence of LQR gains $\{K_{j}\}_{j=1}^{100}$, corresponding to the linearized dynamics about each trim. For the closed-loop nonlinear system, the gain matrices at other state vectors are linearly interpolated over $\{K_{j}\}_{j=1}^{100}$. As shown in Fig. \ref{EigenLQR} (b), depending on the choice of the trim conditions corresponding to the grid-points in scheduling subspace, some open-loop linearized plants are unstable but all closed-loop synthesis are stable.


\section{Probabilistic Robustness Analysis: An Optimal Transport Framework}

\subsection{Closed-loop Uncertainty Propagation}
We assume that the uncertainties in initial conditions ($x_{0}$) and parameters ($p$) are described by the initial joint PDF $\varphi_{0}\left(x_{0}, p\right)$, and this PDF is known for the purpose of performance analysis. For $t>0$, under the action of the closed-loop dynamics, $\varphi_{0}$ evolves over the \emph{extended state space}, defined as the joint space of states and parameters, to yield the instantaneous joint PDF $\varphi\left(x(t), p, t\right)$. Although the closed-loop dynamics governing the state evolution is nonlinear, the {Perron-Frobenius (PF) operator} \cite{LasotaMackey}, governing the joint PDF evolution remains linear. This enables \emph{meshless} computation of $\varphi\left(x(t), p, t\right)$ in \emph{exact arithmetic}, as detailed below.

\subsubsection{Liouville PDE formulation}
The transport equation associated with the PF operator, governing the spatio-temporal evolution of probability mass over the extended state space $\widetilde{x} := \left[x(t), \quad p\right]^{\top}$, is given by the stochastic Liouville equation
\begin{eqnarray}
\displaystyle\frac{\partial \varphi}{\partial t} = -\displaystyle\sum_{i=1}^{n_{x}+n_{p}} \displaystyle\frac{\partial}{\partial \widetilde{x}_{i}}\left(\varphi \: \widetilde{f}_{\text{cl}}\right), \qquad x\left(t\right) \in \mathbb{R}^{n_{x}}, \; p \in \mathbb{R}^{n_{p}},
\label{LiouvillePDE}
\end{eqnarray}
where $\widetilde{f}_{\text{cl}}\left(x\left(t\right), p, t\right)$ denotes the closed-loop extended vector field, i.e.
\begin{eqnarray}
\widetilde{f}_{\text{cl}}\left(x\left(t\right), p, t\right) := \left[\underbrace{f_{\text{cl}}\left(x\left(t\right), p, t\right)}_{n_{x}\times 1}, \quad \underbrace{\mathbf{0}}_{n_{p}\times 1}\right]^{\top}, \qquad \dot{x} = f_{\text{cl}}\left(x(t), p, t\right).
\label{FinalClosedLoopDynamics}
\end{eqnarray}
Since (\ref{LiouvillePDE}) is a first-order PDE, it allows method-of-characteristics (MOC) formulation, which we describe next.

\subsubsection{Characteristic ODE computation}
It can be shown \cite{HalderBhattacharya2011JGCD} that the characteristic curves for (\ref{LiouvillePDE}), are the trajectories of the closed-loop ODE $\dot{x} = f_{\text{cl}}\left(x\left(t\right), p, t\right)$. If the nonlinear vector field $f_{\text{cl}}$ is Lipschitz, then the trajectories are unique, and hence the characteristic curves are non-intersecting. Thus, instead of solving the PDE boundary value problem (\ref{LiouvillePDE}), we can solve the following initial value problem \cite{HalderBhattacharya2011JGCD,HalderBhattacharya2010GNC}:
\begin{eqnarray}
\dot{\varphi} = - \left(\nabla \cdot f_{\text{cl}}\right) \: \varphi, \qquad \varphi\left(x_{0}, p, 0\right) = \varphi_{0},
\label{CharacteristicODE}
\end{eqnarray}
along the trajectories $x\left(t\right)$. Notice that solving (\ref{CharacteristicODE}) along one trajectory, is independent of the other, and hence the formulation is a natural fit for parallel implementation. This computation differs from Monte Carlo (MC) as shown in Table \ref{MCvsPF}.
\begin{table}
 \begin{center}
\caption{Comparison of joint PDF computation over $\mathbb{R}^{n_{x}+n_{p}}$: MC vs. PF}
\label{MCvsPF}
  \begin{tabular}{|l|l|l|}\hline
\textbf{Attributes} & \textbf{MC simulation} & \textbf{PF via MOC}\\ \hline\hline
Concurrency & Offline post-processing & Online \\
Accuracy & Histogram approximation & Exact arithmetic\\
Spatial discretization & Grid based & Meshless \\
ODEs per sample & $n_{x}$ & $n_{x}+1$\\\hline
 \end{tabular}
 \end{center}
\end{table}

We emphasize here that the MOC solution of Liouville equation is a Lagrangian (as opposed to Eulerian) computation and hence, has no residue or equation error. The latter would appear if we directly employ function approximation techniques to numerically solve the Liouville equation (see e.g. \cite{pantano2007least}). Here, instead we non-uniformly sample the \emph{known} initial PDF via Markov Chain Monte Carlo (MCMC) technique \cite{diaconis2009markov} and then co-integrate the state and density value at that state location over time. Thus, the numerical accuracy is as good as the temporal integrator.

Further, there is no loss of generality in this \emph{finite} sample computation. If at any fixed time $t>0$, one seeks to evaluate the joint PDF value at an arbitrary location $\widetilde{x}^{\star}(t)$ in the extended state space, then one could back-propagate $\widetilde{x}^{\star}(t)$ via the given dynamics till $t=0$, resulting $\widetilde{x}_{0}^{\star}$. Intuitively, $\widetilde{x}_{0}^{\star}$ signifies the initial condition from which the query point $\widetilde{x}^{\star}(t)$ could have come. If $\widetilde{x}_{0}^{\star} \in \text{supp}\left(\varphi_{0}\right)$, then we forward integrate (\ref{CharacteristicODE}) with $\widetilde{x}_{0}^{\star}$ as the initial condition, to determine joint PDF value at $\widetilde{x}^{\star}(t)$. If $\widetilde{x}_{0}^{\star} \notin \text{supp}\left(\varphi_{0}\right)$, then $\widetilde{x}^{\star}(t)\notin \text{supp}\left(\varphi\left(\widetilde{x}(t),t\right)\right)$, and hence the joint PDF value at $\widetilde{x}^{\star}(t)$ would be zero.

Notice that the divergence computation in (\ref{CharacteristicODE}) can be done analytically \emph{offline} for our case of LQR and gsLQR closed-loop systems, provided we obtain function approximations for aerodynamic coefficients. However, there are two drawbacks for such offline computation of the divergence. \textbf{First}, the accuracy of the computation will depend on the quality of function approximations for aerodynamic coefficients. \textbf{Second}, for nonlinear controllers like MPC \cite{camacho2004model}, which numerically realize the state feedback, analytical computation for closed-loop divergence is not possible. For these reasons, we implement an alternative \emph{online} computation of divergence in this paper. Using the Simulink\textsuperscript{\textregistered} command \texttt{linmod}, we linearize the \emph{closed-loop systems} about each characteristics, and obtain the instantaneous divergence as the trace of the time-varying Jacobian matrix. Algorithm \ref{AlgoUncProp} details this method for closed-loop uncertainty propagation. Specific simulation set up for our F-16 closed-loop dynamics is given in Section V.A.

\begin{algorithm}[h]
\caption{Closed-loop Uncertainty Propagation via MOC Solution of Liouville PDE}
\label{AlgoUncProp}
\begin{algorithmic}[1]
\footnotesize{\Require The initial joint PDF $\varphi_{0}\left(x_{0},p\right)$, closed-loop dynamics (\ref{FinalClosedLoopDynamics}), number of samples $N$, final time $t_{f}$, time step $\Delta t$.
\State Generate $N$ scattered samples $\{x_{0i},p_{i}\}_{i=1}^{N}$ from the initial PDF $\varphi_{0}\left(x_{0},p\right)$ \Comment{Using MCMC}
\State Evaluate the samples $\{x_{0i},p_{i}\}_{i=1}^{N}$ at $\varphi_{0}\left(x_{0},p\right)$, to get the point cloud $\{x_{0i},p_{i},\varphi_{0i}\}_{i=1}^{N}$
\For {$t = 0 :\Delta t : t_{f}$} \Comment{Index for time}
\For {$i = 1 : 1 : N$} \Comment{Index for samples}
\State Numerically integrate the closed-loop dynamics (\ref{FinalClosedLoopDynamics}) \Comment{Propagate states to obtain $\{x_{i}(t)\}_{i=1}^{N}$}
\State Compute $\nabla \cdot f_{\text{cl}}$ using Simulink\textsuperscript{\textregistered} command \texttt{linmod} \Comment{Since divergence at $i$\textsuperscript{th} sample at time $t$ = trace of}
\State $\qquad\qquad\qquad\qquad\qquad\qquad\qquad\qquad\qquad\qquad\qquad\qquad\qquad$ Jacobian of $f_{\text{cl}}$, evaluated at $x_{i}(t)$
\State Numerically integrate the characteristic ODE (\ref{CharacteristicODE}) \Comment{Propagate joint PDF values to get $\{\varphi_{i}(t) \triangleq \varphi\left(x_{i}(t), p_{i}, t\right)\}_{i=1}^{N}$}
\EndFor
\EndFor \Comment{We get time-varying probability-weighted scattered data $\{x_{i}(t), p_{i}, \varphi_{i}(t)\}_{i=1}^{N}$ for each time $t$}
}
\end{algorithmic}
\end{algorithm}

\subsection{Optimal Transport to Trim}

\subsubsection{Wasserstein metric}
To provide a quantitative comparison for LQR and gsLQR controllers' performance, we need a notion of ``distance" between the respective time-varying state PDFs and the desired state PDF. Since the controllers strive to bring the state trajectory ensemble to $x_{\text{trim}}$, hence we take $\varphi^{*}\left(x_{\text{trim}}\right)$, a Dirac delta distribution at $x_{\text{trim}}$, as our desired joint PDF. The notion of distance must compare the concentration of trajectories in the state space and for meaningful inference, should define a metric. Next, we describe \emph{Wasserstein metric}, that meets these axiomatic requirements \cite{HalderBhattacharyaCDC2011} of ``distance" on the manifold of PDFs.

\subsubsection{Definition}
Consider the metric space $\left(M,\ell_{2}\right)$ and take $y, \widehat{y} \in M$. Let $\mathcal{P}_{2}\left(M\right)$ denote the collection of all probability measures $\mu$ supported on $M$, which have finite 2\textsuperscript{nd} moment. Then the $L_{2}$ Wasserstein distance of order $2$, denoted as $_{2}W_{2}$, between two probability measures $\varsigma, \widehat{\varsigma} \in \mathcal{P}_{2}\left(M\right)$, is defined as
\begin{align}
_{2}W_{2}\left(\varsigma, \widehat{\varsigma}\right) := \left(\underset{\mu \in \mathcal{M}\left(\varsigma,\widehat{\varsigma}\right)}{\text{inf}} \displaystyle\int_{M \times M} \parallel y - \widehat{y} \parallel_{\ell_{2}}^{2} \: d\mu\left(y,\widehat{y}\right) \right)^{\frac{1}{2}}
\label{WassDef}
\end{align}
where $\mathcal{M}\left(\varsigma,\widehat{\varsigma}\right)$ is the set of all measures supported on the product space $M \times M$, with first marginal $\varsigma$, and second marginal $\widehat{\varsigma}$.

Intuitively, Wasserstein distance equals the \emph{least amount of work} needed to convert one distributional shape to the other, and can be interpreted as the cost for Monge-Kantorovich optimal transportation plan \cite{Villani2003}. The particular choice of $L_{2}$ norm with order 2 is motivated in \cite{HalderBhattacharyaCDC2012}. For notational ease, we henceforth denote $_{2}W_{2}$ as $W$. One can prove (p. 208, \cite{Villani2003}) that $W$ defines a metric on the manifold of PDFs.

\subsubsection{Computation of $W$}

In general, one needs to compute $W$ from its definition, which requires solving a \emph{linear program} (LP) \cite{HalderBhattacharyaCDC2011} as follows. Recall that the MOC solution of the Liouville equation, as explained in Section IV.A, results time-varying scattered data. In particular, at any \emph{fixed} time $t>0$, the MOC computation results scattered sample points $\mathcal{Y}_{t} := \{y_{i}\}_{i=1}^{m}$ over the state space, where each sample $y_{i}$ has an associated joint probability mass function (PMF) value $\varsigma_{i}$. If we sample the reference PDF likewise and let $\widehat{\mathcal{Y}}_{t} := \{\widehat{y}_{i}\}_{i=1}^{n}$, then computing $W$ between the instantaneous and reference PDF reduces to computing (\ref{WassDef}) between two sets of scattered data: $\{y_{i},\varsigma_{i}\}_{i=1}^{m}$ and $\{\widehat{y}_{j},\widehat{\varsigma}_{j}\}_{j=1}^{n}$. Further, if we interpret the squared inter-sample distance $c_{ij}:=\parallel y_{i} - \widehat{y}_{j} \parallel_{\ell_{2}}^{2}$ as the cost of transporting unit mass from location $y_{i}$ to $\widehat{y}_{j}$, then according to (\ref{WassDef}), computing $W^{2}$ translates to
\begin{eqnarray}
\text{minimize} \; \displaystyle\sum_{i=1}^{m}\displaystyle\sum_{j=1}^{n} \: c_{ij} \: \mu_{ij}
\label{HitchcockKoopmansLP}
\end{eqnarray}
subject to the constraints
\begin{equation*}
\displaystyle\sum_{j=1}^{n} \mu_{ij} = \varsigma_{i}, \qquad \forall \; y_{i} \in \mathcal{Y}_{t},
\tag{C1}
\end{equation*}
\begin{equation*}
\displaystyle\sum_{i=1}^{m} \mu_{ij} = \widehat{\varsigma}_{j}, \qquad \forall \; \widehat{y}_{j} \in \widehat{\mathcal{Y}}_{t},
\tag{C2}
\end{equation*}
\begin{equation*}
\qquad\qquad\quad\;\;\mu_{ij} \geqslant 0, \qquad\; \forall \; \left(y_{i},\widehat{y}_{j}\right) \in \mathcal{Y}_{t} \times \widehat{\mathcal{Y}}_{t}.
\tag{C3}
\end{equation*}
In other words, the objective of the LP is to come up with an optimal mass transportation policy $\mu_{ij} := \mu\left(y_{i} \rightarrow \widehat{y}_{j}\right)$ associated with cost $c_{ij}$. Clearly, in addition to constraints (C1)--(C3), (\ref{HitchcockKoopmansLP}) must respect the necessary feasibility condition
\begin{equation*}
\displaystyle\sum_{i=1}^{m} \varsigma_{i} = \displaystyle\sum_{j=1}^{n} \widehat{\varsigma}_{j},
\tag{C0}
\end{equation*}
denoting the conservation of mass. In our context of measuring the shape difference between two PDFs, we treat the joint PMF vectors $\varsigma_{i}$ and $\widehat{\varsigma}_{j}$ to be the marginals of some unknown joint PMF $\mu_{ij}$ supported over the product space $\mathcal{Y}_{t} \times \widehat{\mathcal{Y}}_{t}$. Since determining joint PMF with given marginals is not unique, (\ref{HitchcockKoopmansLP}) strives to find that particular joint PMF which minimizes the total cost for transporting the probability mass while respecting the normality condition.

Notice that, (\ref{HitchcockKoopmansLP}) is an LP in $m n$ variables, subject to $\left(m + n + mn\right)$ constraints, with $m$ and $n$ being the cardinality of the respective scattered data representation of the PDFs under comparison. As shown in \cite{HalderBhattacharyaCDC2012}, the main source of computational burden in solving this LP, stems from \emph{storage complexity}. It is easy to verify that the \emph{sparse constraint matrix representation} requires $\left(6mn + \left(m+n\right)d + m + n\right)$ amount of storage, while the same for \emph{non-sparse representation} is $\left(m+n\right)\left(mn + d + 1\right)$, where $d$ is the dimension of the support for each PDF. Notice that $d$ enters linearly through $\ell_{2}$ norm computation, but the storage complexity grows \emph{polynomially} with $m$ and $n$. We observed that with sparse LP solver MOSEK \cite{mosek2010mosek}, on a standard computer with 4 GB memory, one can go up to $m=n=$ 3000 samples. On the other hand, increasing the number of samples, increases the accuracy \cite{HalderBhattacharyaCDC2012} of finite-sample $W$ computation. This leads to numerical accuracy versus storage capacity trade off.

\subsubsection{Reduction of storage complexity}

For our purpose of computing $W\left(\varphi\left(x\left(t\right),t\right), \varphi^{*}\left(x_{\text{trim}}\right)\right)$, the storage complexity can be reduced by leveraging the fact that $\varphi^{*}\left(x_{\text{trim}}\right)$ is a stationary Dirac distribution. Hence, it suffices to represent the joint probability mass function (PMF) of $\varphi^{*}\left(x_{\text{trim}}\right)$ as a single sample located at $x_{\text{trim}}$ with PMF value unity. This trivializes the optimal transport problem, since
\begin{eqnarray}
W\left(t\right) \triangleq W\left(\varphi\left(x\left(t\right),t\right), \varphi^{*}\left(x_{\text{trim}}\right)\right)
= \sqrt{\displaystyle\sum_{i=1}^{n} \parallel x_{i}\left(t\right) - x_{\text{trim}} \parallel_{2}^{2} \: \gamma_{i}},
\label{TrivialTransport}
\end{eqnarray}
where $\gamma_{i} \geqslant 0$ denotes the joint PMF value at sample $x_{i}\left(t\right)$, $i=1,\hdots,n$. Consequently, the storage complexity reduces to $\left(nd + n + d\right)$, which is \emph{linear} in number of samples.

\section{Numerical Results}

\subsection{Robustness Against Initial Condition Uncertainty}

\subsubsection{Stochastic initial condition uncertainty}
We first consider analyzing the controller robustness subject to initial condition uncertainties. For this purpose, we let the initial condition $x_{0}$ to be a stochastic perturbation from $x_{\text{trim}}$, i.e. $x_{0} = x_{\text{trim}} + x_{\text{pert}}$, where $x_{\text{pert}}$ is a random vector with probability density $\varphi_{\text{pert}} = \mathcal{U}\left(\left[\theta^{\text{min}}_{\text{pert}},\theta^{\text{max}}_{\text{pert}}\right] \times \left[V^{\text{min}}_{\text{pert}},V^{\text{max}}_{\text{pert}}\right] \times \left[\alpha^{\text{min}}_{\text{pert}},
\alpha^{\text{max}}_{\text{pert}}\right] \times \left[q^{\text{min}}_{\text{pert}},q^{\text{max}}_{\text{pert}}\right] \right)$, where the perturbation range for each state, is listed in Table \ref{xPertRange}. Consequently, $x_{0}$ has a joint PDF $\varphi_{0}\left(x_{0}\right)$. For this analysis, we assume no actuator disturbance.
\begin{table}
 \begin{center}
\caption{Admissible state perturbation limits}
\label{xPertRange}
  \begin{tabular}{|l|l|}\hline
\textbf{$x_{\text{pert}}$} & \textbf{Interval}\\ \hline\hline
$\theta_{\text{pert}} \in \left[\theta^{\text{min}}_{\text{pert}},\theta^{\text{max}}_{\text{pert}}\right]$ & $\left[-35^{\circ}, + 35^{\circ}\right]$ \\
$V_{\text{pert}} \in \left[V^{\text{min}}_{\text{pert}},V^{\text{max}}_{\text{pert}}\right]$ & $\left[-\text{65 ft/s}, + \text{65 ft/s}\right]$\\
$\alpha_{\text{pert}} \in \left[\alpha^{\text{min}}_{\text{pert}},\alpha^{\text{max}}_{\text{pert}}\right]$ & $\left[-20^{\circ}, + 50^{\circ}\right]$\\
$q_{\text{pert}} \in \left[q^{\text{min}}_{\text{pert}},q^{\text{max}}_{\text{pert}}\right]$ & $\left[-\text{70 deg/s}, +\text{70 deg/s}\right]$\\\hline
 \end{tabular}
 \end{center}
\end{table}

\subsubsection{Simulation set up}

We generated pseudo-random Halton sequence \cite{Niederreiter1992} in $\left[\theta^{\text{min}}_{\text{pert}},\theta^{\text{max}}_{\text{pert}}\right] \times \left[V^{\text{min}}_{\text{pert}},V^{\text{max}}_{\text{pert}}\right] \times \left[\alpha^{\text{min}}_{\text{pert}},\alpha^{\text{max}}_{\text{pert}}\right] \times \left[q^{\text{min}}_{\text{pert}},q^{\text{max}}_{\text{pert}}\right]$, to sample the uniform distribution $\varphi_{\text{pert}}$, and hence $\varphi_{0}$ supported on the four dimensional state space. With 2000 Halton samples for $\varphi_{0}$, we propagate joint state PDFs for both LQR and gsLQR closed-loop dynamics via MOC ODE (\ref{CharacteristicODE}), from $t=0$ to 20 seconds, using fourth-order Runge-Kutta integrator with fixed step-size $\Delta t = 0.01$ s.

We observed that the \texttt{linmod} computation for evaluating time-varying divergence along each trajectory, takes the most of computational time. To take advantage of the fact that computation along characteristics are independent of each other, all simulations were performed using 12 cores with NVIDIA\textsuperscript{\textregistered} Tesla GPUs in MATLAB\textsuperscript{\textregistered} environment. It was noticed that with LQR closed-loop dynamics, the computational time for single sample from $t=0$ to $20$ s, is approx. 90 seconds. With sequential \texttt{for}-loops over 2000 samples, this scales to 50 hours of runtime. The same for gsLQR scales to 72 hours of runtime. In parallel implementation on Tesla, MATLAB\textsuperscript{\textregistered} \texttt{parfor}-loops were used to reduce these runtimes to 4.5 hours (for LQR) and 6 hours (for gsLQR), respectively.

\subsubsection{Density based qualitative analysis}
Fig. \ref{1dmarg} shows the evolution of univariate \emph{marginal error PDFs}. All marginal computations were performed using
algorithms previously developed by the authors \cite{HalderBhattacharya2011JGCD}. Since $\varphi_{0}$ and its marginals were uniform, Fig. \ref{1dmarg_001} shows similar trend for small $t$, and there seems no visible difference between LQR and gsLQR performance. As $t$ increases, both LQR and gsLQR error PDFs shrink about zero. By $t=20$ s (Fig. \ref{1dmarg_20}), both LQR and gsLQR controllers make the respective state marginals $\varphi^{j}(t)$, $j=1,\hdots,4$, converge to the Dirac distribution at $x^{j}_{\text{trim}}$, although the rate of convergence of gsLQR error marginals is faster than the same for LQR.
\begin{figure}[tb]
\begin{center}
\centering
\subfigure[Snapshot at $t = 0.01$ second.]
{\hspace*{-0.1in}\includegraphics[scale=0.36]{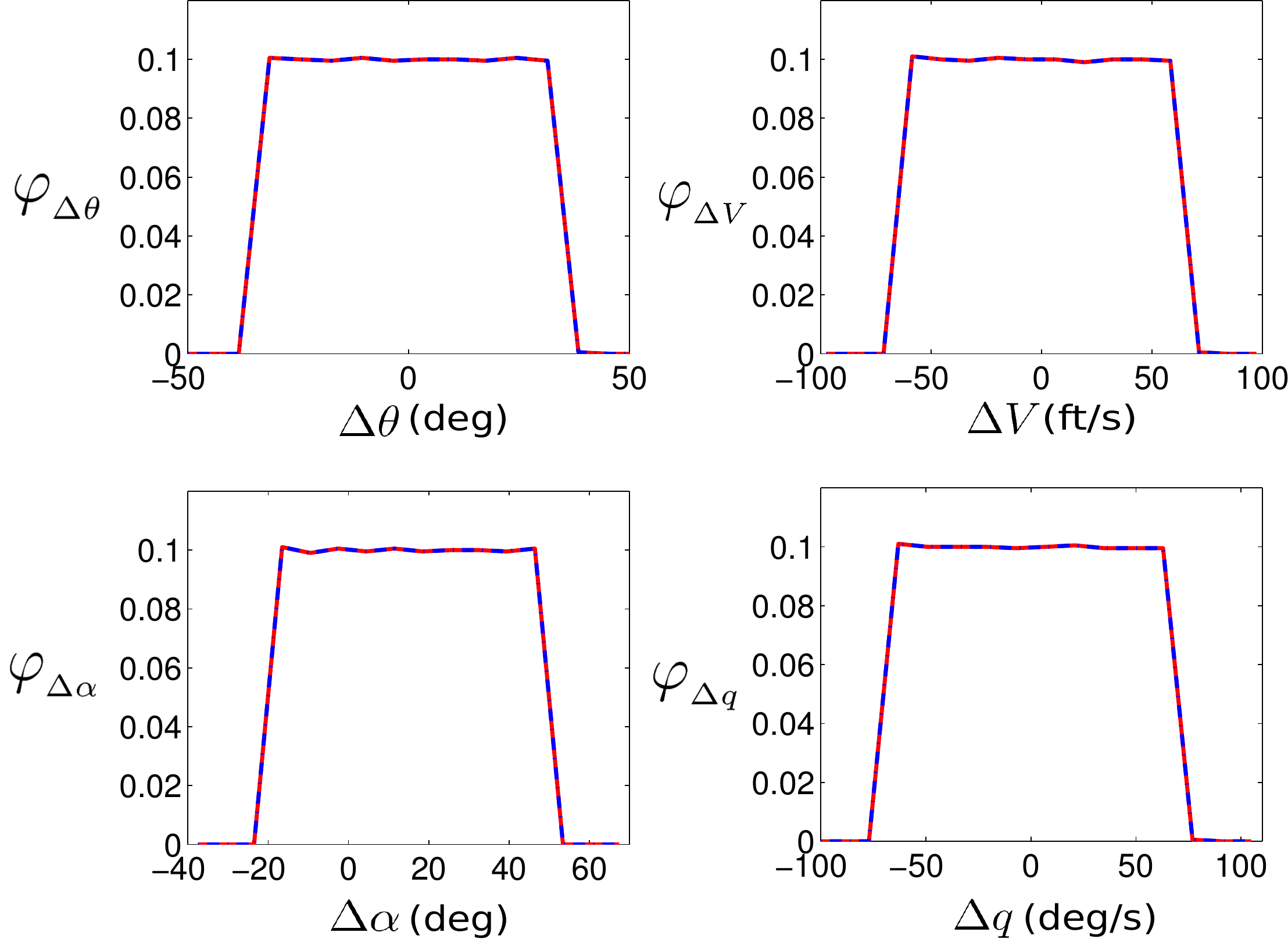}\label{1dmarg_001}}
\subfigure[Snapshot at $t = 1$ second.]
{\hspace*{0.1in}\includegraphics[scale=0.36]{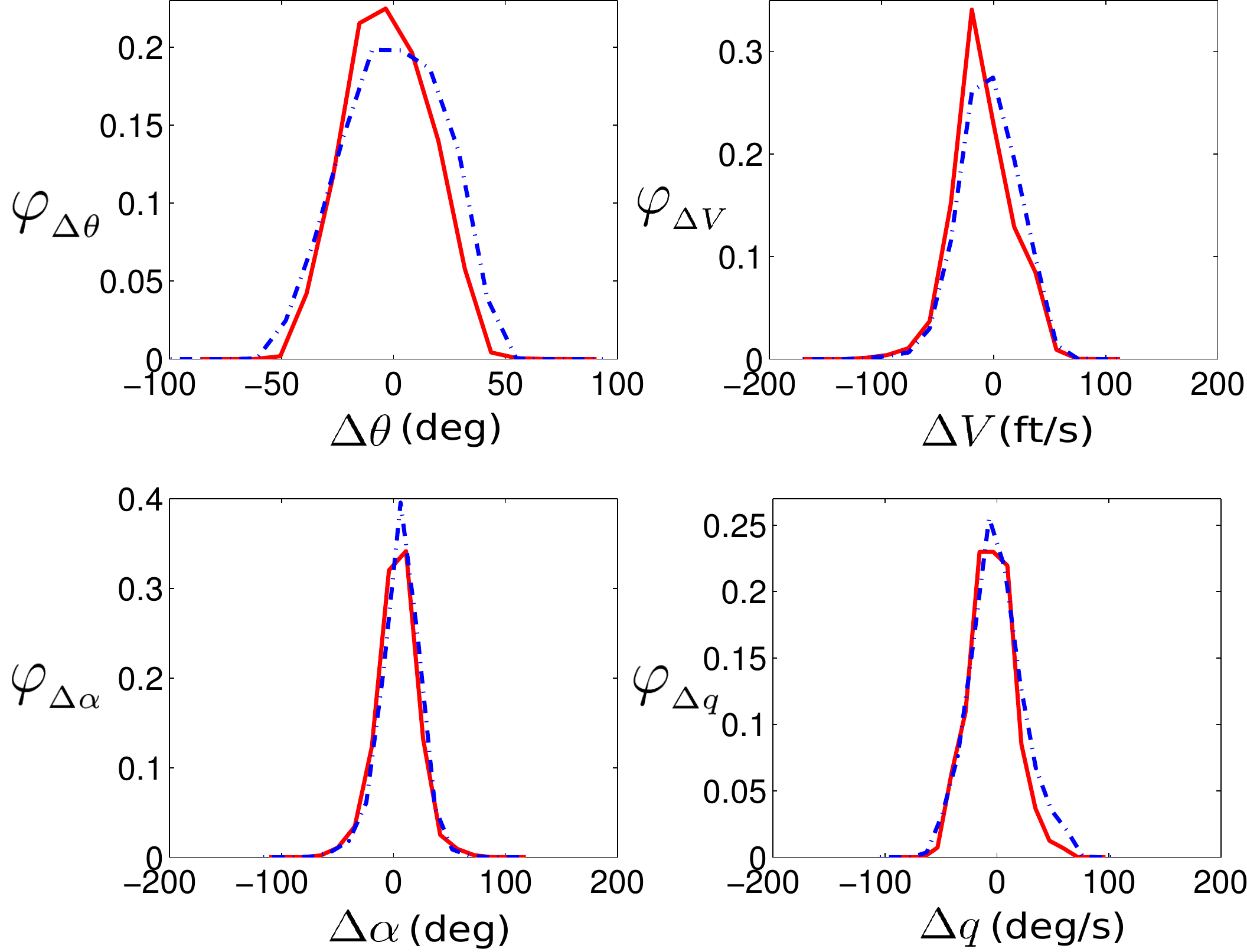}\label{1dmarg_1}}
\subfigure[Snapshot at $t = 5$ seconds.]
{\hspace*{-0.1in}\includegraphics[scale=0.36]{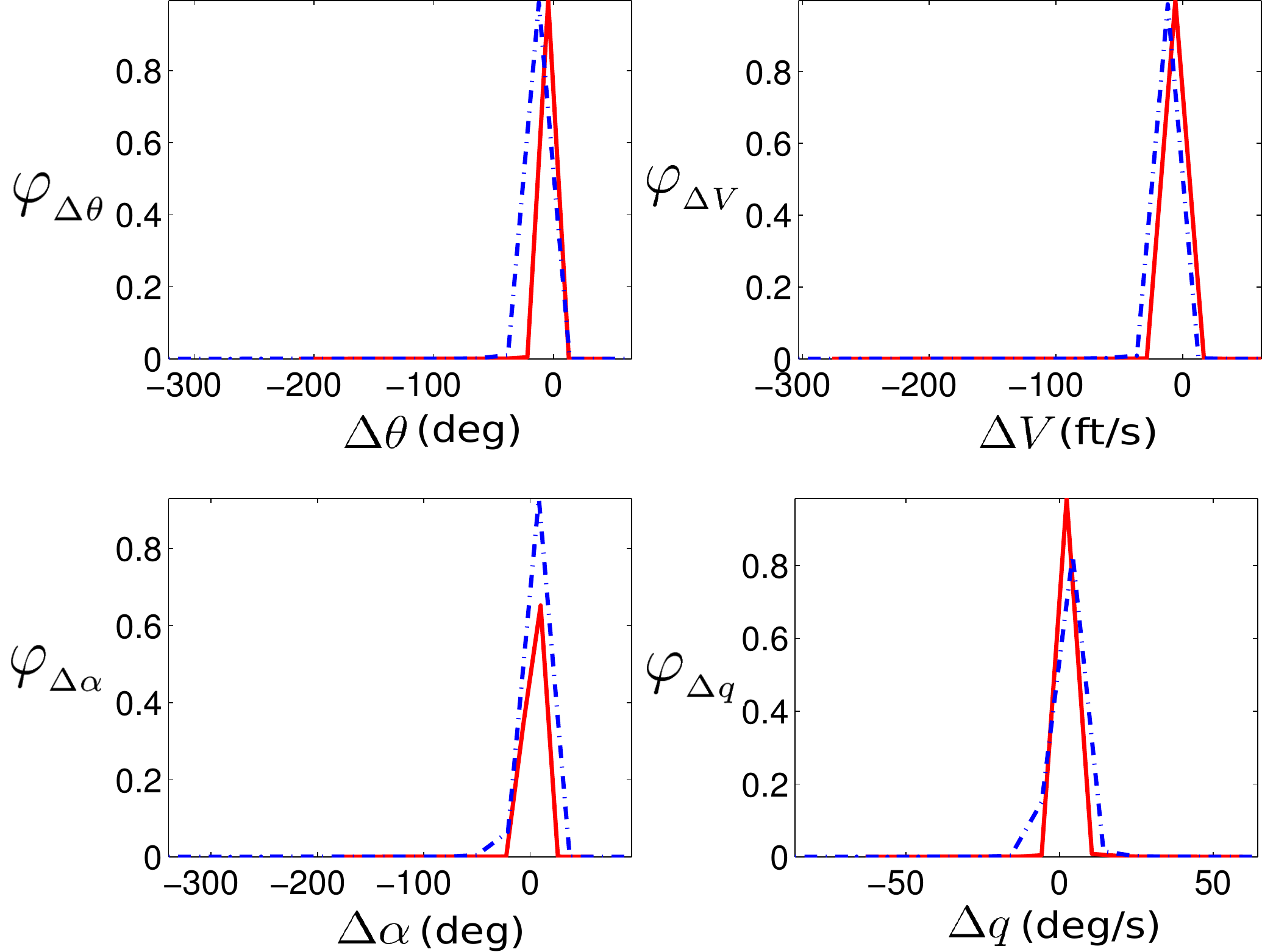}\label{1dmarg_5}}
\subfigure[Snapshot at $t = 20$ seconds.]
{\hspace*{0.1in}\includegraphics[scale=0.36]{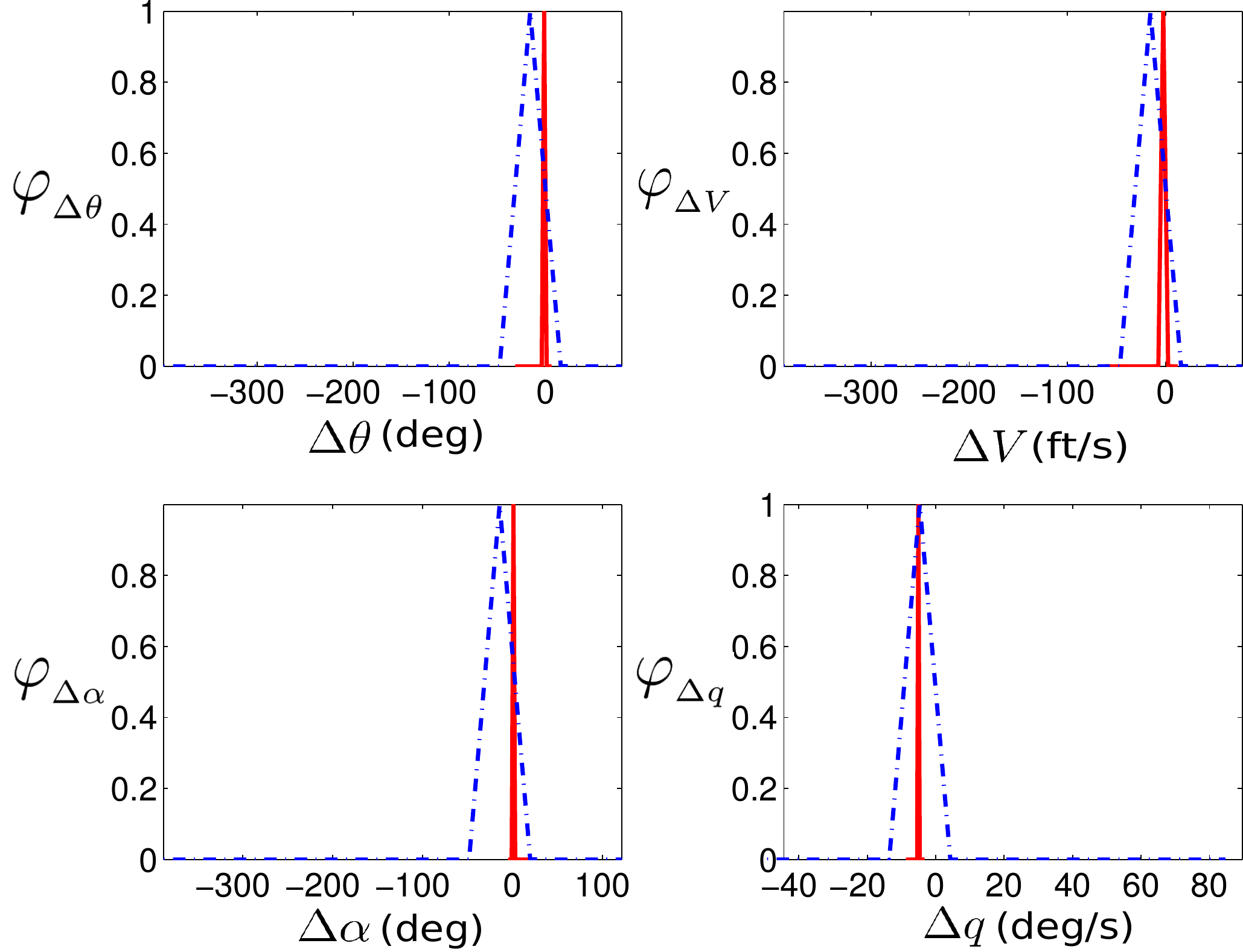}\label{1dmarg_20}}
\end{center}
\caption{Snapshots of univariate marginal error PDFs for each state, with LQR \textit{(blue, dashed)} and gsLQR \textit{(red, solid)} closed loop dynamics.}\label{1dmarg}
\end{figure}
Thus, Fig. \ref{1dmarg} \emph{qualitatively} show that both LQR and gsLQR exhibit comparable \emph{immediate} and
\emph{asymptotic} performance, as far as robustness against initial condition uncertainty is concerned. However,
there are some visible mismatches in Fig. \ref{1dmarg_1} and \ref{1dmarg_5}, that suggests the need for a careful quantitative
investigation of the transient performance.

The insights obtained from Fig. \ref{1dmarg} can be verified against the MC simulations (Fig. \ref{MonteCarlo}). Compared to LQR, the MC simulations reveal faster regulation performance for gsLQR, and hence corroborate the faster rate of convergence of gsLQR error marginals observed in Fig. \ref{1dmarg}. From Fig. \ref{MonteCarlo}, it is interesting to observe that by $t = 20$ s, some of the LQR trajectories do not converge to trim while all gsLQR trajectories do. For risk aware control design, it is natural to ask: how probable is this event, i.e. can we probabilistically assess the severity of the loss of performance for LQR?
\begin{figure}[tb]
\vspace*{-0.3in}
\begin{center}
\centering
\subfigure[State error vs. time for LQR controller]{\hspace*{-0.2in}\includegraphics[scale=.26]{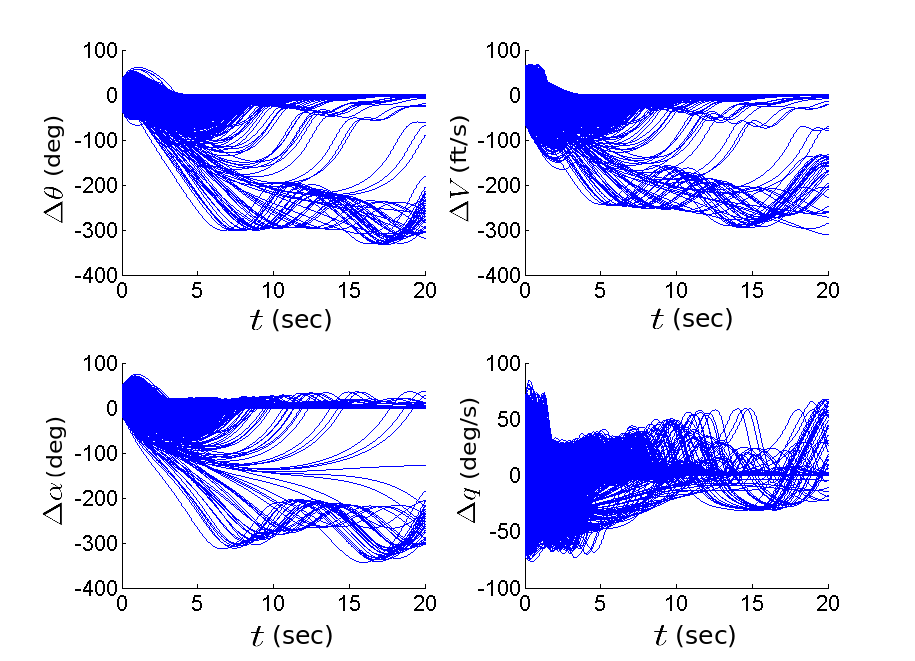}\label{lqr_prop}}\vspace*{0.15in}
\subfigure[State error vs. time for gsLQR controller]{\hspace*{-0.02in}\includegraphics[scale=.26]{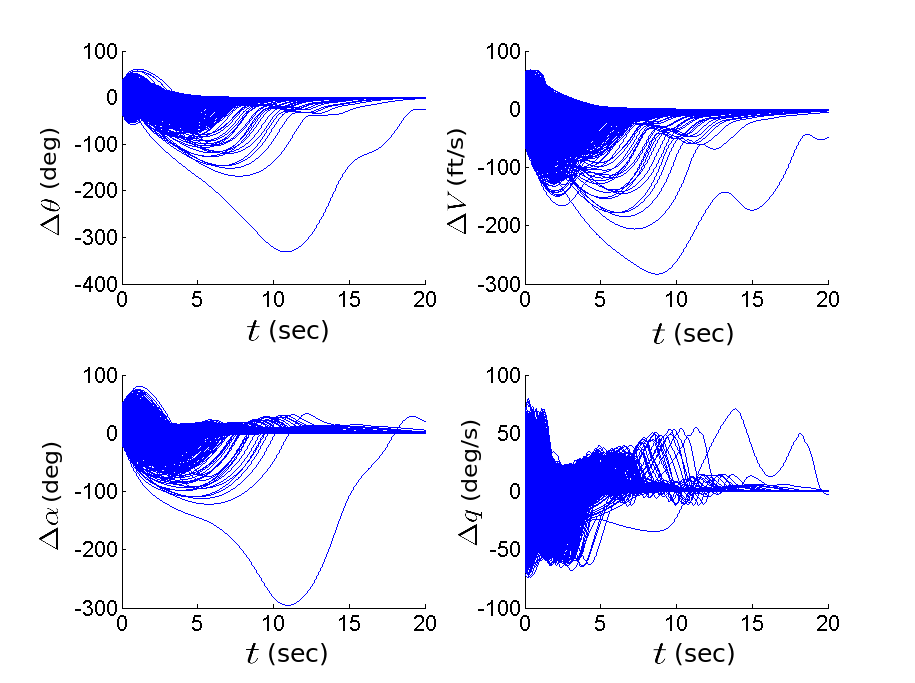}\label{gslqr_prop}}
\end{center}
\vspace*{-0.2in}
\caption{MC state error ($\Delta x^{j}\left(t\right) \triangleq x^{j}\left(t\right) - x^{j}_{\text{trim}}$, $j=1,\hdots,4$) trajectories for LQR and gsLQR closed-loop dynamics.}\label{MonteCarlo}
\end{figure}
To address this question, in Fig. \ref{logrho}, we plot the time evolution of the peak value of LQR joint state PDF, and compare that with the joint state PDF values along the LQR closed-loop trajectories that don't converge to $x_{\text{trim}}$ by $20$ s. Fig. \ref{logrho} reveals that the probabilities that the LQR trajectories don't converge, remain \emph{at least} an order of magnitude less than the peak value of the LQR joint PDF. In other words, the performance degradation for LQR controller, as observed in Fig. \ref{lqr_prop}, is a low-probability event. This conclusion can be further verified from Fig. \ref{MinMax}, which shows that for gsLQR controller, both maximum and minimum probability trajectories achieve satisfactory regulation performance by $t= 20$ s. However, for LQR controller, although the maximum probability trajectory achieves regulation performance as good as the corresponding gsLQR case, the minimum probability LQR trajectory results in poor regulation. Furthermore, even for the maximum probability trajectories (Fig. \ref{MinMax}, \emph{top row}), there are transient performance mismatch between LQR and gsLQR, for approximately 3--8 s.

\begin{figure}[t]
\begin{minipage}[b]{0.5\linewidth}
\centering
\vspace*{-0.01in}
\includegraphics[width=2.9in]{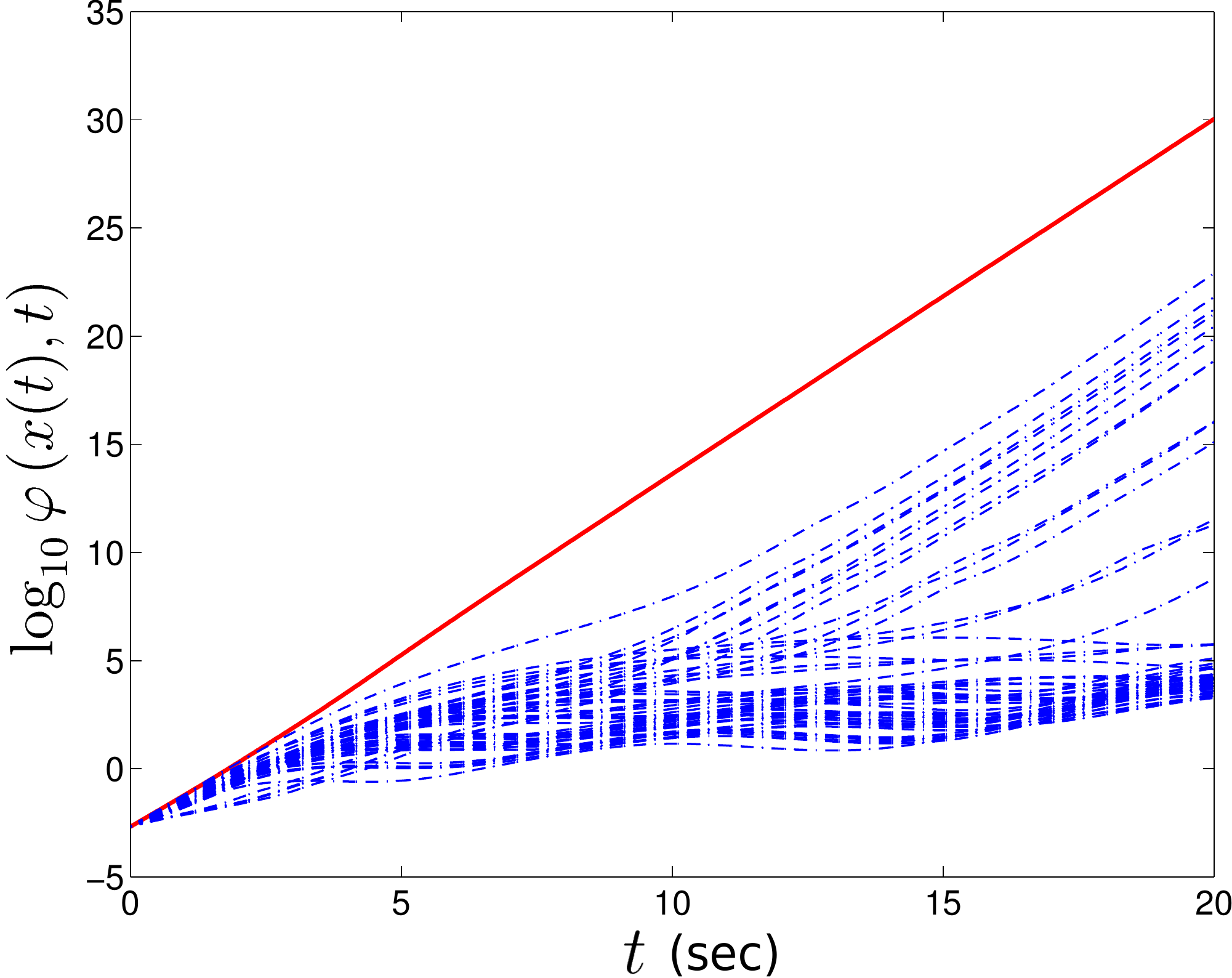}
\caption{Time evolution of \emph{maximum value} of joint PDF $\varphi_{\text{LQR}}\left(x(t), t\right)$ (\textit{red solid}) and $\varphi_{\text{LQR}}\left(x(t), t\right)$ \emph{along the diverging trajectories} (\textit{blue dashed}), as seen in Fig. \ref{lqr_prop}. The plots are in \emph{log-linear} scale.}
\label{logrho}
\end{minipage}
\hspace{0.2cm}
\begin{minipage}[b]{0.5\linewidth}
\centering
\includegraphics[width=3.2in]{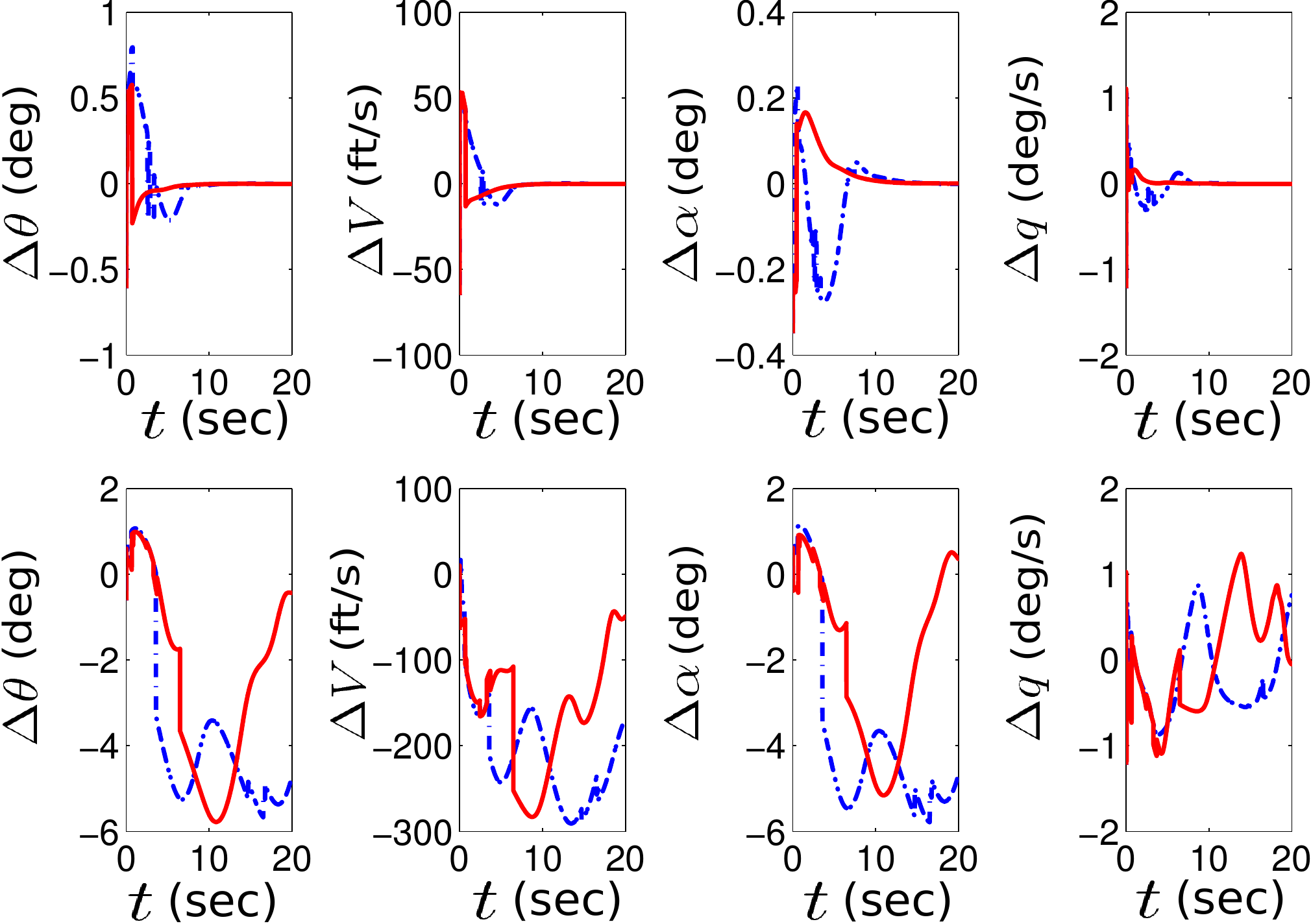}
\caption{Time evolution of the most likely (\emph{top row}) and least likely (\emph{bottom row}) state errors for LQR (\emph{blue dashed}) and gsLQR (\emph{red solid}) closed-loop dynamics.}
\label{MinMax}
\end{minipage}
\end{figure}

\subsubsection{Optimal transport based quantitative analysis}
From a systems-control perspective, instead of performing an \emph{elaborate qualitative} statistical analysis as above, one would like to have a \emph{concise and quantitative} robustness analysis tool, enabling the inferences of Section V.A.2 and V.A.3. We now illustrate that the Wasserstein distance introduced in Section IV.B, serves this purpose.

In this formulation, a controller is said to have better regulation performance if it makes the closed-loop state PDF converge faster to the Dirac distribution located at $x_{\text{trim}}$. In other words, for a better controller, at all times, the distance between the closed-loop state PDF and the Dirac distribution, as measured in $W$, must remain smaller than the same for the other controller. Thus, we compute the time-evolution of the two Wasserstein distances
\begin{eqnarray}
 W_{\text{LQR}}\left(t\right) &\triangleq& W\left(\varphi_{\text{LQR}}\left(x\left(t\right), t\right),\varphi^{*}\left(x_{\text{trim}}\right)\right),\\
 W_{\text{gsLQR}}\left(t\right) &\triangleq& W\left(\varphi_{\text{gsLQR}}\left(x\left(t\right), t\right), \varphi^{*}\left(x_{\text{trim}}\right)\right).
 \end{eqnarray}
The schematic of this computation is shown in Fig. \ref{LQRgsLQRcomparison}.
\begin{figure}[h]
\begin{center}
\centering
\hspace*{-0.43in}
\includegraphics[scale=0.42]{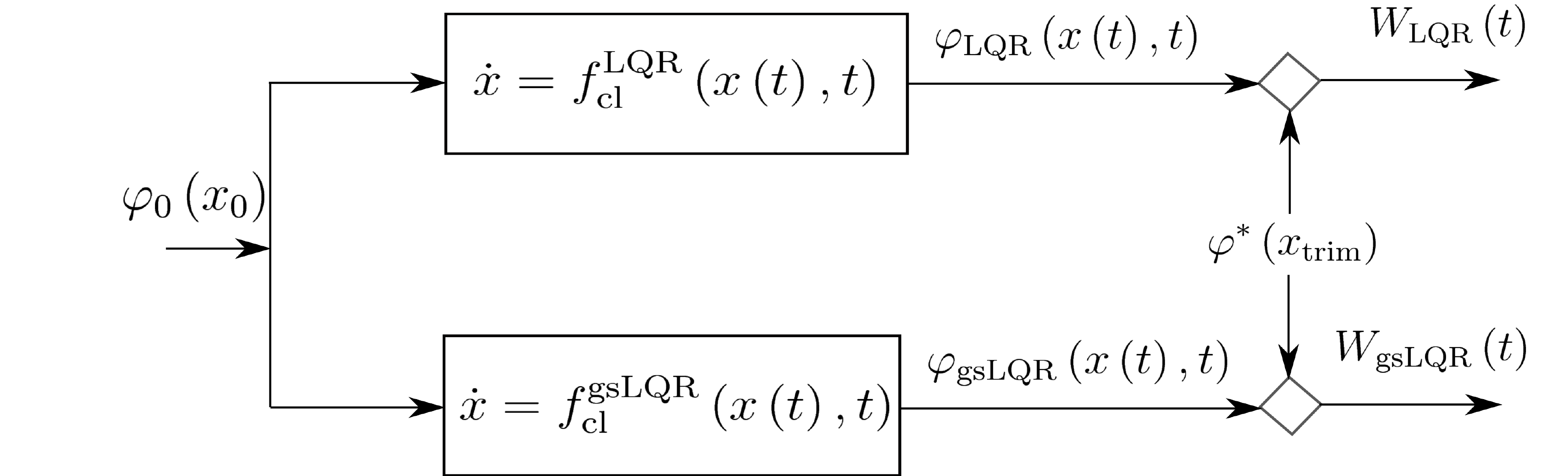}
\end{center}
\caption{Schematic of probabilistic robustness comparison for controllers based on Wasserstein metric. The ``diamond" denotes the Wasserstein computation by solving the Monge-Kantorovich optimal transport. The internal details of LQR and gsLQR closed-loop dynamics blocks are as in Fig. \ref{ClosedLoopBlockDiagrams}.}
\label{LQRgsLQRcomparison}
\end{figure}

Fig. \ref{wdist} indeed confirms the qualitative trends, observed in the density based statistical analysis mentioned before, that LQR and gsLQR exhibit comparable immediate and asymptotic performance. Furthermore, Fig. \ref{wdist} shows that for $t=3-8$ seconds, $W_{\text{LQR}}$ stays higher than $W_{\text{gsLQR}}$, meaning the gsLQR joint PDF $\varphi_{\text{gsLQR}}\left(x(t), t\right)$ is closer to $\varphi^{*}\left(x_{\text{trim}}\right)$, compared to the LQR joint PDF $\varphi_{\text{LQR}}\left(x(t), t\right)$. This corroborates well with the transient mismatch observed in Fig. \ref{1dmarg_5}. As time progresses, both $W_{\text{LQR}}$ and $W_{\text{gsLQR}}$ converge to zero, meaning the convergence of both LQR and gsLQR closed-loop joint state PDFs to the Dirac distribution at $x_{\text{trim}}$.
\begin{figure}[thb]
\begin{center}
\centering
\hspace*{-0.1in}\includegraphics[scale=0.42]{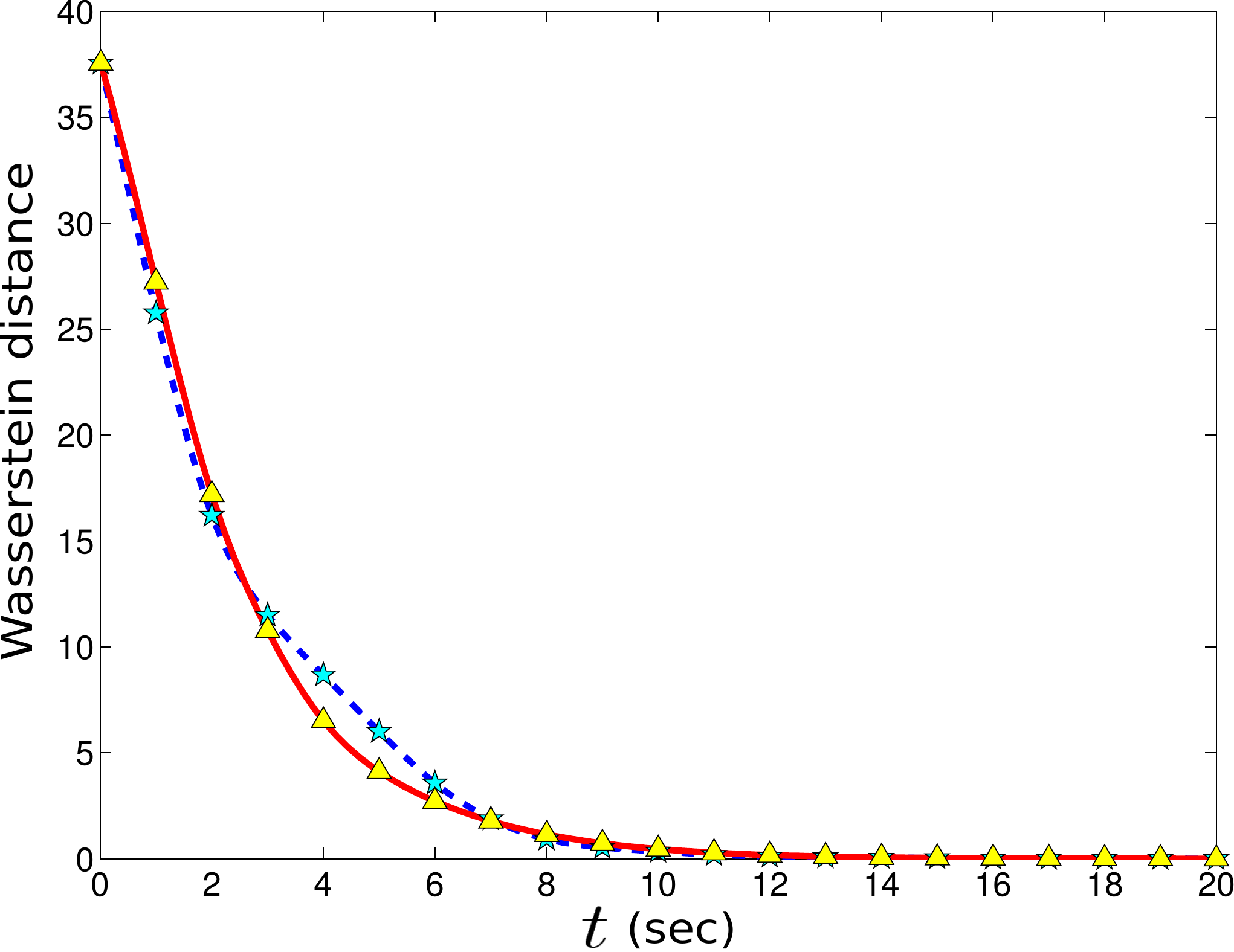}
\vspace*{-0.1in}
\end{center}
\caption{Comparison of time histories of $W\left(\varphi_{\text{LQR}}(t), \varphi^{*}\right)$ \textit{(blue dashed, star)} and $W\left(\varphi_{\text{gsLQR}}(t), \varphi^{*}\right)$ \textit{(red solid, triangle)}.}
\label{wdist}
\end{figure}

\begin{remark}
At this point, we highlight a subtle distinction between the two approaches of probabilistic robustness analysis presented above: (1) \emph{density based qualitative analysis}, and (2) \emph{the optimal transport based quantitative analysis using Wasserstein distance}. For density based qualitative analysis, controller performance assessment was done using Fig. \ref{1dmarg} that compares the asymptotic convergence of the univariate marginal state PDFs. However, this analysis is only \emph{sufficient} since convergence of marginals does not \emph{necessarily} imply convergence of joints. Conversely, the optimal transport based quantitative analysis is \emph{necessary and sufficient} since Fig. \ref{wdist} compares the Wasserstein distance between the joint PDFs. We refer the readers to Appendix A for a precise statement and proof.

Further, since $W_{\text{LQR}}\left(t\right) \rightarrow 0$ for large $t$, we can affirmatively say that the divergent LQR trajectories are indeed of low-probability, as hinted by Fig. \ref{logrho} and \ref{MinMax}. Otherwise, $W_{\text{LQR}}\left(t\right)$ would show a steady-state error. Thus, the Wasserstein distance is shown to be an effective way of comparing the robustness of controllers.
\end{remark}

\subsection{Robustness Against Parametric Uncertainty}

\subsubsection{Deterministic initial condition with stochastic parametric uncertainty}
Instead of the stochastic initial condition uncertainties described in Section V.A.1, we now consider uncertainties in three parameters: mass of the aircraft ($m$), true $x$-position of c.g. ($x_{\text{cg}}$), and pitch moment-of-inertia ($J_{yy}$). The uncertainties in these geometric parameters can be attributed to the variable rate of fuel consumption depending on the flight conditions. For the simulation purpose, we assume that each of these three parameters has $\pm \Delta \%$ \emph{uniform} uncertainties about their nominal values listed in Table \ref{ParamF16}. To verify the controller robustness, we vary the parametric uncertainty range by allowing $\Delta = 0.5, 2.5, 5, 7.5$ and $15$. As before, we set the actuator disturbance $w=0$.

\subsubsection{Simulation set up}
We let the initial condition be a deterministic vector: $x_{0} = x_{\text{trim}} + x_{\text{pert}}$, where $x_{\text{pert}} = \left[1.1803\:\text{rad}, \; 5.1058\:\text{ft/s}, \; 2.8370\:\text{rad}, \; 10^{-4}\:\text{rad/s}\right]^{\top}$. We keep the rest of the simulation set up same as in the previous case. Notice that since $\dot{p} = 0$, the characteristic ODE for joint PDF evolution remains the same. However, the state trajectories, along which the characteristic ODE needs to be integrated, now depend on the realizations of the random vector $p$.

\subsubsection{Density based qualitative analysis}
Due to parametric uncertainties in $p \triangleq \left[m, \; x_{\text{cg}}, \; J_{yy}\right]^{\top}$, we now have $n_{x} = 4, n_{p} = 3$, and hence the joint PDF evolves over the extended state space $\widetilde{x}(t) \triangleq \left[x(t), \; p\right]^{\top} \in \mathbb{R}^{7}$. Since we assumed $x_{0}$ to be deterministic, both initial and asymptotic joint PDFs $\varphi_{0}$ and $\varphi_{\infty}$ are degenerate distributions, supported over the three dimensional parametric subspace of the extended state space in $\mathbb{R}^{7}$. In other words, $\varphi_{0} = \varphi_{p}\left(p\right) \delta\left(x - x_{0}\right)$, and $\varphi_{\infty} = \varphi_{p}\left(p\right) \delta\left(x - x_{\text{trim}}\right)$, i.e. the PDFs $\varphi_{0}$ and $\varphi_{\infty}$ differ only by a translation of magnitude $\parallel x_{0} - x_{\text{trim}}\parallel_{2} = \parallel x_{\text{pert}} \parallel_{2}$. However, for any intermediate time $t \in \left(0,\infty\right)$, the joint PDF $\varphi\left(\widetilde{x}(t), t\right)$ has a support obtained by nonlinear transformation of the initial support. This is illustrated graphically in Fig. \ref{ParamUncDegen}.

\begin{figure}[tb]
\begin{center}
\centering
\includegraphics[width=\textwidth]{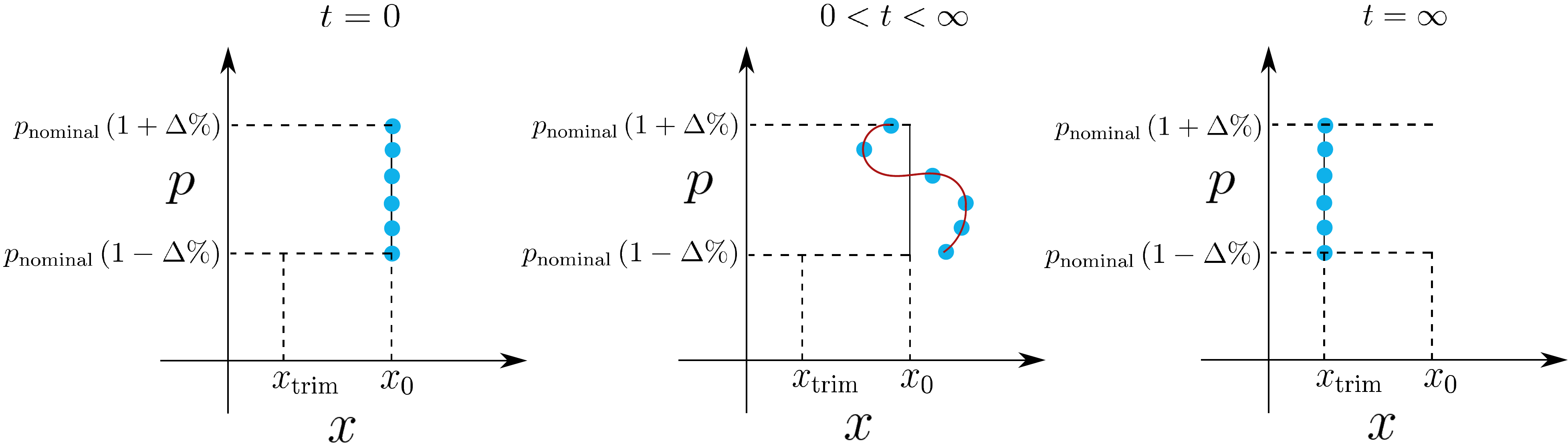}
\end{center}
\vspace*{-0.2in}
\caption{A schematic of how the support of a joint PDF evolves in the extended state space under parametric uncertainty. For ease of understanding, we illustrate here a case for one state $x$ and one parameter $p$. Since $x_{0}$ is deterministic but $p$ is random, the initial joint PDF $\varphi_{0}$ is simply the univariate parametric PDF $\varphi_{p}(p)$ translated to $x=x_{0}$. Consequently, $\varphi_{0}$ is supported on a straight line segment (one dimensional subspace) in the two dimensional extended state space, as shown in the \emph{left figure}. For $0<t<\infty$, due to state dynamics, the samples (denoted as \emph{circles}) on that line segment move in the horizontal ($x$) direction while keeping the respective ordinate ($p$) value constant, resulting the instantaneous support to be a curve (\emph{middle figure}). If the system achieves regulation, then $\lim_{t\rightarrow\infty} x(t) = x_{\text{trim}}$, $\forall p$ in the parametric uncertainty set, resulting the asymptotic joint PDF $\varphi_{\infty}$ to be supported on a straight line segment (\emph{right figure}) at $x=x_{\text{trim}}$.}
\label{ParamUncDegen}
\end{figure}

The MC simulations in Fig. \ref{MonteCarloParamUnc} show that both LQR and gsLQR have similar asymptotic performance, however, the transient overshoot for LQR is much larger than the same for gsLQR. Hence, the transient performance for gsLQR seems to be more robust against parametric uncertainties. Similar trends were observed for other values of $\Delta$.

\begin{figure}[t]
\vspace*{-0.3in}
\begin{center}
\centering
\subfigure[State error vs. time for LQR controller]{\hspace*{-0.2in}\includegraphics[scale=.26]{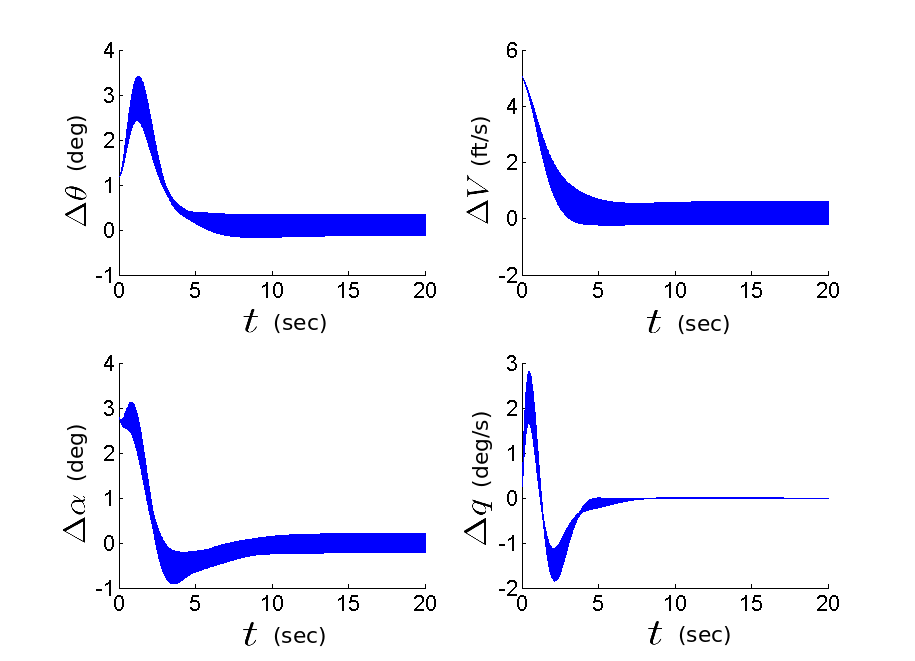}\label{lqr_param5perc}}\vspace*{0.15in}
\subfigure[State error vs. time for gsLQR controller]{\hspace*{-0.02in}\includegraphics[scale=.26]{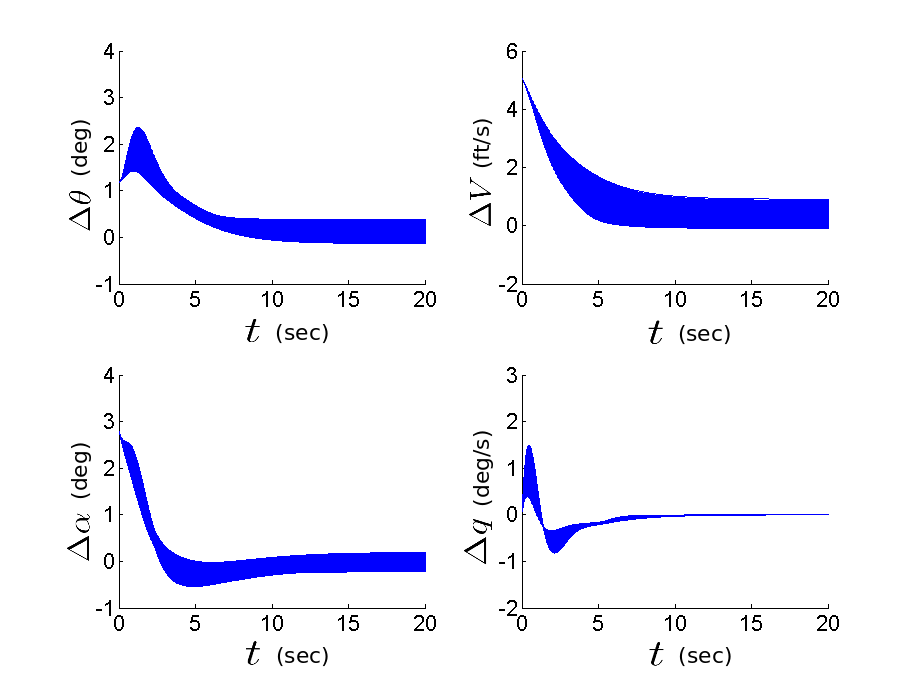}\label{gslqr_param5perc}}
\end{center}
\caption{MC state error ($\Delta x^{j}\left(t\right) \triangleq x^{j}\left(t\right) - x^{j}_{\text{trim}}$, $j=1,\hdots,4$) trajectories for LQR and gsLQR closed-loop dynamics, with $\pm 2.5\%$ uniform uncertainties in $p = \left[m, \; x_{\text{cg}}, \; J_{yy}\right]^{\top}$, i.e. $p = p_{\text{nominal}}\left(1 \pm \Delta\%\right)$, where $\Delta = 2.5$, and $p_{\text{nominal}}$ values are listed in Table \ref{ParamF16}.}\label{MonteCarloParamUnc}
\end{figure}

\subsubsection{Optimal transport based quantitative analysis}
Here, we solve the LP (\ref{HitchcockKoopmansLP}) with cost
\begin{eqnarray}
c_{ij} = n \displaystyle\sum_{i=1}^{n}\parallel x_{i}(t) - x_{\text{trim}} \parallel_{2}^{2} + \displaystyle\sum_{i=1}^{n}\displaystyle\sum_{j=1}^{n}\displaystyle\sum_{k=1}^{n_{p}=3} \left(p_{k}\left(i\right) - p_{k}\left(j\right)\right)^{2},
\label{ParamLPcost}
\end{eqnarray}
with $\varsigma_{i}$ being the joint PMF value at the $i$\textsuperscript{th} sample location $\widetilde{x}_{i}(t) = \left[x_{i}(t), \; p(i)\right]^{\top}$, and $\widehat{\varsigma}_{j}$ being the trim joint PMF value at the $j$\textsuperscript{th} sample location $\left[x_{\text{trim}}, \; p(j)\right]^{\top}$. Fig. \ref{WlqrParamUnc} and \ref{WgslqrParamUnc} show $W\left(t\right)$ vs. $t$ under parametric uncertainty for LQR and gsLQR, respectively. For both the controllers, the plots confirm that larger parametric uncertainty results in larger transport efforts at all times, causing higher value of $W$. In both cases, the deterministic (no uncertainty) $W$ curves (\emph{dashed lines} in Fig. \ref{WlqrParamUnc} and \ref{WgslqrParamUnc}) almost coincide with those of $\pm 0.5\%$ parametric uncertainties. Notice that in the deterministic case, $W$ is simply the Euclidian distance of the current state from trim, i.e. convergence in $W$ reduces to the classical $\ell_{2}$ convergence of a signal.

\begin{figure}[htb]
\begin{minipage}[b]{0.5\linewidth}
\centering
\includegraphics[width=3.2in]{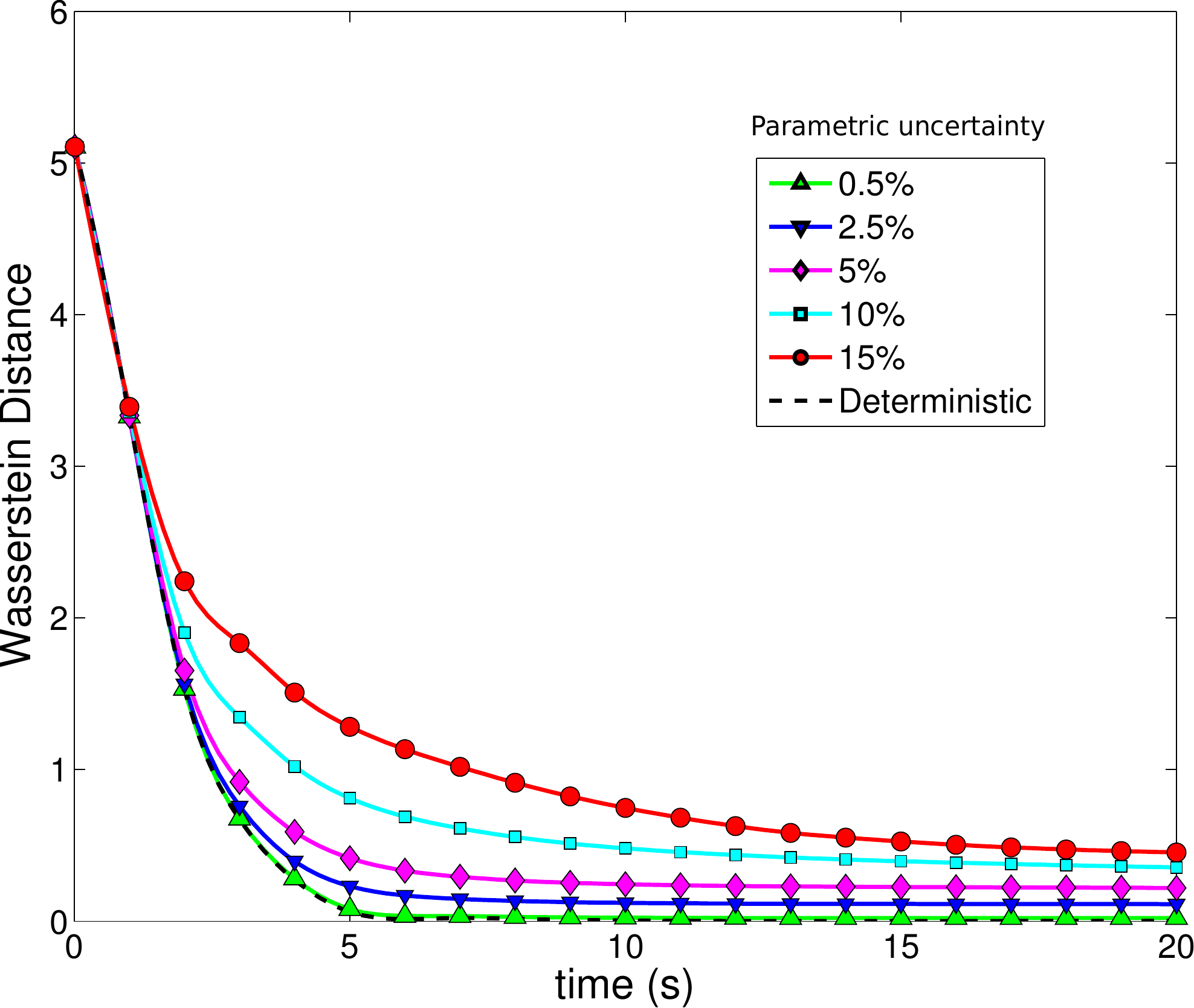}
\caption{Time evolution of Wasserstein distance for LQR, with varying levels of $\Delta$.}
\label{WlqrParamUnc}
\end{minipage}
\hspace{0.2cm}
\begin{minipage}[b]{0.5\linewidth}
\centering
\includegraphics[width=3.2in]{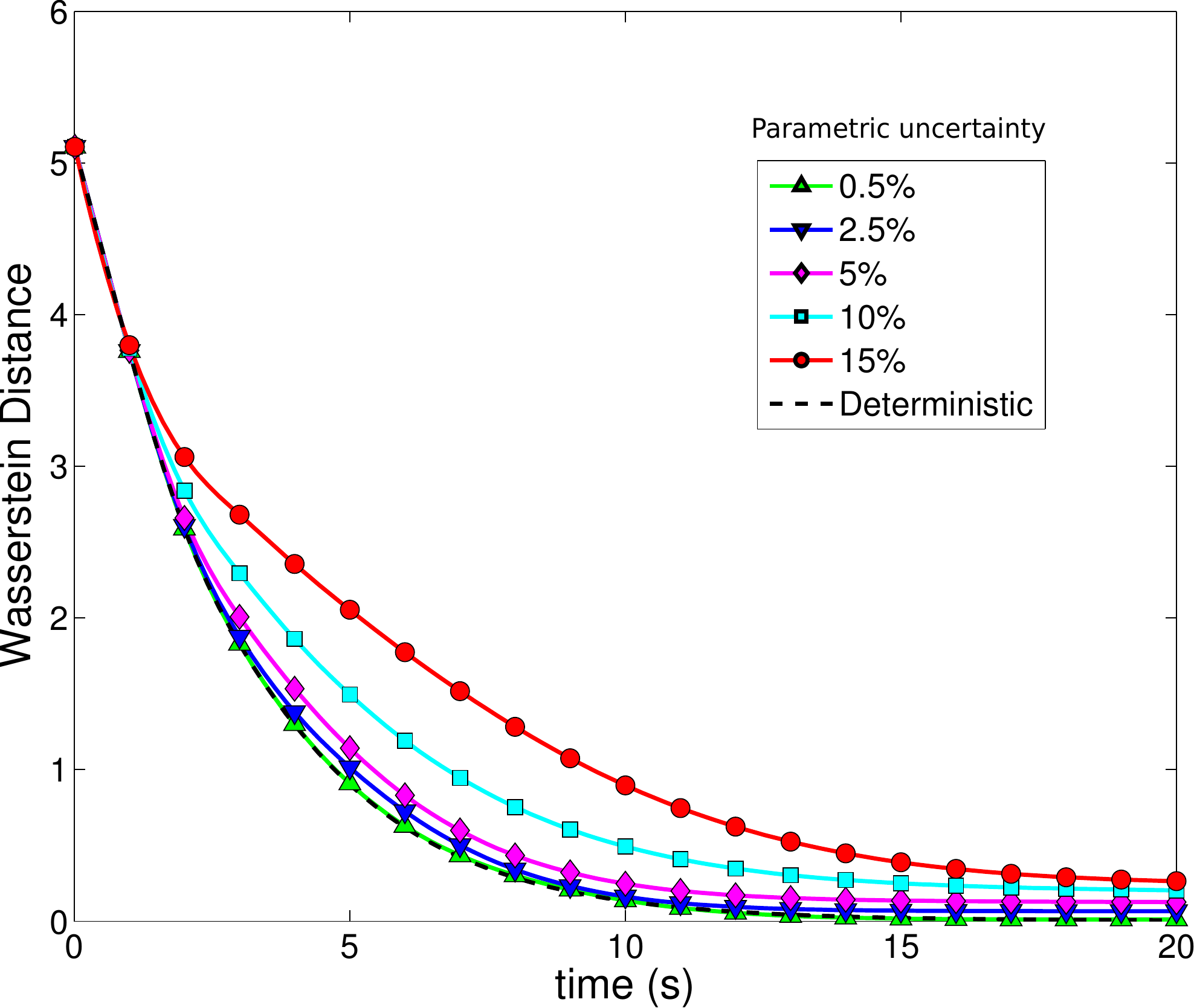}
\caption{Time evolution of Wasserstein distance for gsLQR, with varying levels of $\Delta$.}
\label{WgslqrParamUnc}
\end{minipage}
\end{figure}

It is interesting to compare the LQR and gsLQR performance against parametric uncertainty for each fixed $\Delta$. For $0-3$ s, the rate-of-convergence for $W(t)$ is faster for LQR, implying probabilistically faster regulation. However, the LQR $W$ curves tend to flatten out after 3 s, thus slowing down its joint PDF's rate-of-convergence to $\varphi_{\infty}$. On the other hand, gsLQR $W$ curves exhibit somewhat opposite trend. The initial regulation performance for gsLQR is slower than that of LQR, but gsLQR achieves better \emph{asymptotic} performance by bringing the probability mass closer to $\varphi_{\infty}$ than the LQR case, resulting smaller values of $W$. Further, one may notice that for large $(\pm 15\%)$ parametric uncertainties, the $W$ curve for LQR shows a mild bump around 3 s, corresponding to the significant transient overshoot observed in Fig. \ref{lqr_param5perc}. This can be contrasted with the corresponding $W$ curve for gsLQR, that does not show any prominent effect of transient overshoot at that time. The observation is consistent with the MC simulation results in Fig. \ref{gslqr_param5perc}. Thus, we can conclude that gsLQR is more robust than LQR, against parametric uncertainties.

\subsection{Robustness Against Actuator Disturbance}

\subsubsection{Stochastic initial condition uncertainty with actuator disturbance}
Here, in addition to the initial condition uncertainties described in Section V.A.1, we consider actuator disturbance in elevator. Our objective is to analyze how the additional disturbance in actuator affects the regulation performance of the controllers.

\subsubsection{Simulation set up}
We let the initial condition uncertainties to be described as in Table \ref{xPertRange}, and consequently the initial joint PDF is uniform. Further, we assume that the elevator is subjected to a periodic disturbance of the form $w(t) = 6.5 \sin\left(\Omega t\right)$. The simulation results of Section V.A.1 corresponds to the special case when the forcing angular frequency $\Omega = 0$. To investigate how $\Omega>0$ alters the system response, we first perform frequency-domain analysis of the LQR closed-loop system, linearized about $x_{\text{trim}}$. Fig. \ref{SingularValueFreqResponse} shows the variation in singular value magnitude (in dB) with respect to frequency (rad/s), for the transfer array from disturbance $w(t)$ to states $x(t)$. This frequency-response plot shows that the \emph{peak frequency} is $\omega^{\star} \approx 2$ rad/s.

\begin{figure}[tb]
\begin{center}
\centering
\includegraphics[scale=0.68]{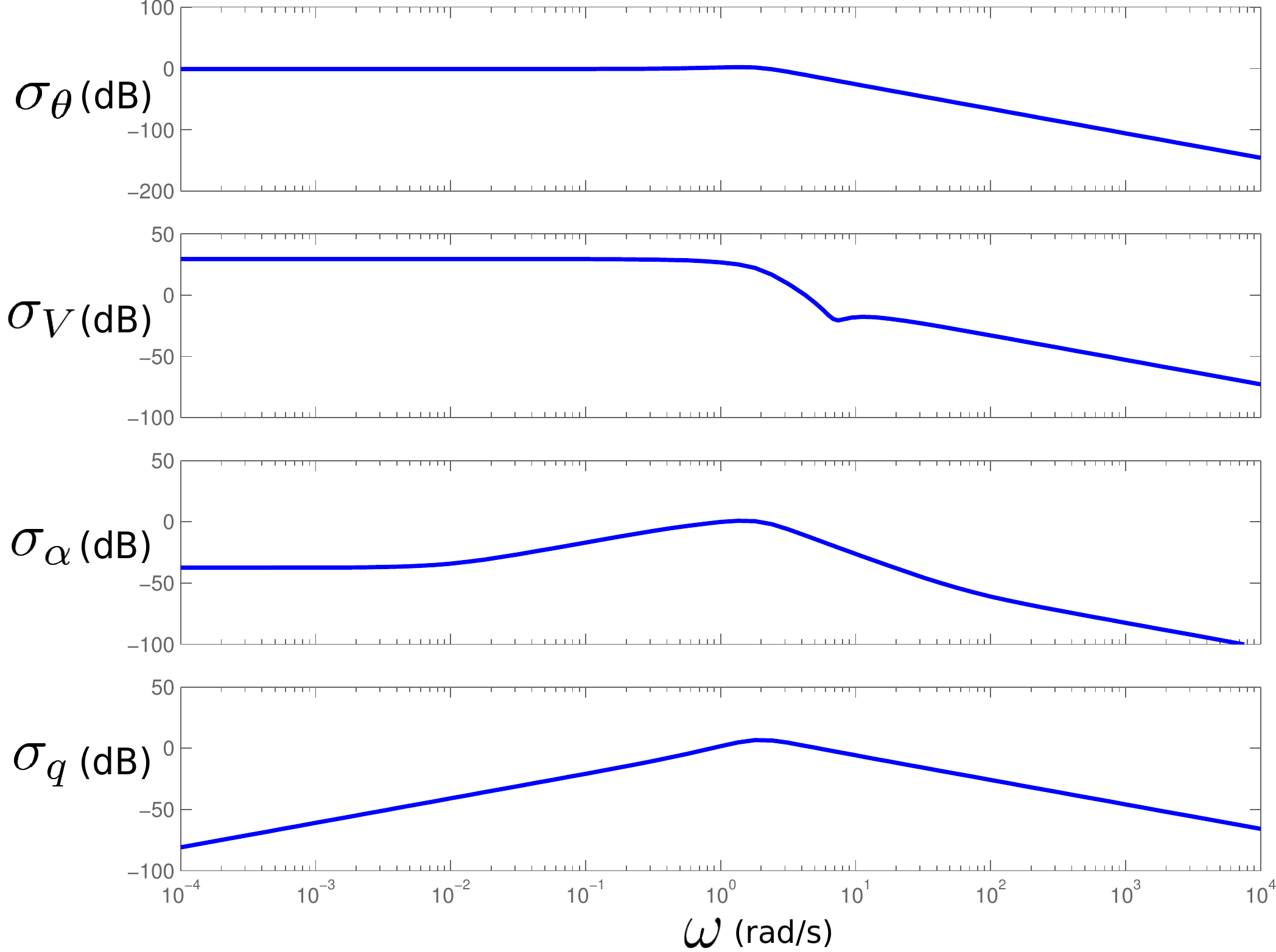}
\end{center}
\vspace*{-0.2in}
\caption{Singular values for the LQR closed-loop dynamics, linearized about $x_{\text{trim}}$, computed from the $4\times1$ transfer array corresponding to the disturbance to states.}
\label{SingularValueFreqResponse}
\end{figure}

\subsubsection{Density based qualitative analysis}
To compare the LQR and gsLQR performance under peak frequency excitation (as per linearized LQR analysis), we set $\Omega = \omega^{\star} = 2$ rad/s, and evolve the initial uniform joint PDF over the LQR and gsLQR closed-loop state space. Notice that the LQR closed-loop dynamics is \emph{nonlinear}, and the extent to which the linear analysis would be valid, depends on the robustness of regulation performance. Fig. \ref{lqr_ActNoise2} shows the LQR state error trajectories from the MC simulation. It can be observed that after $t = 10$ s, most of the LQR trajectories exhibit constant frequency oscillation with $\omega = 2$ rad/s. This trend is even more prominent for the gsLQR trajectories in Fig. \ref{gslqr_ActNoise2}, which seem to settle to the constant frequency oscillation quicker than the LQR case.
\begin{figure}[t]
\vspace*{-0.3in}
\begin{center}
\centering
\subfigure[State error vs. time for LQR controller]{\hspace*{-0.2in}\includegraphics[scale=.26]{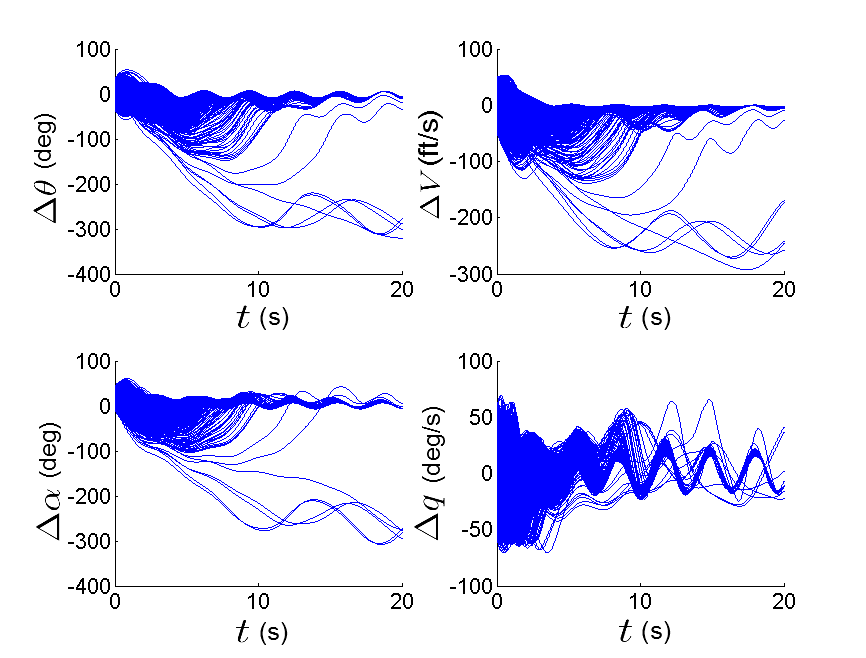}\label{lqr_ActNoise2}}\vspace*{0.15in}
\subfigure[State error vs. time for gsLQR controller]{\hspace*{-0.02in}\includegraphics[scale=.26]{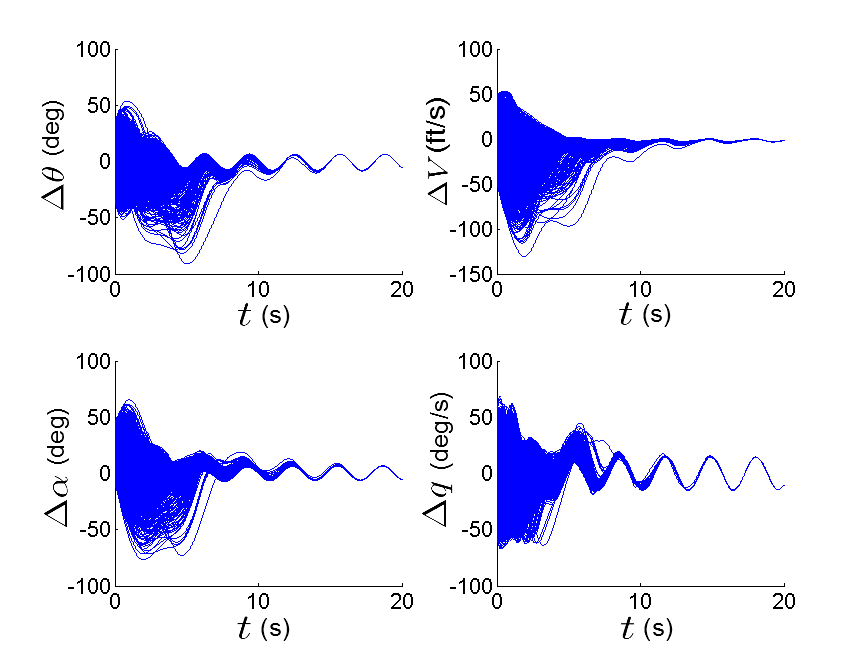}\label{gslqr_ActNoise2}}
\end{center}
\caption{MC state error ($\Delta x^{j}\left(t\right) \triangleq x^{j}\left(t\right) - x^{j}_{\text{trim}}$, $j=1,\hdots,4$) trajectories for LQR and gsLQR closed-loop dynamics, with periodic disturbance $w(t) = 6.5\sin\left(2t\right)$ in the elevator, and initial condition uncertainties.}\label{MonteCarloActuatorNoiseSixPointFiveSineTwot}
\end{figure}

\subsubsection{Optimal transport based quantitative analysis}
We now investigate the effect of elevator disturbance $w(t) = 6.5\sin\left(2t\right)$ and initial condition uncertainties, via the optimal transport framework. In this case, the computation of Wasserstein distance is of the form (\ref{TrivialTransport}).
\begin{figure}[b]
\begin{minipage}[b]{0.5\linewidth}
\centering
\includegraphics[width=3.2in]{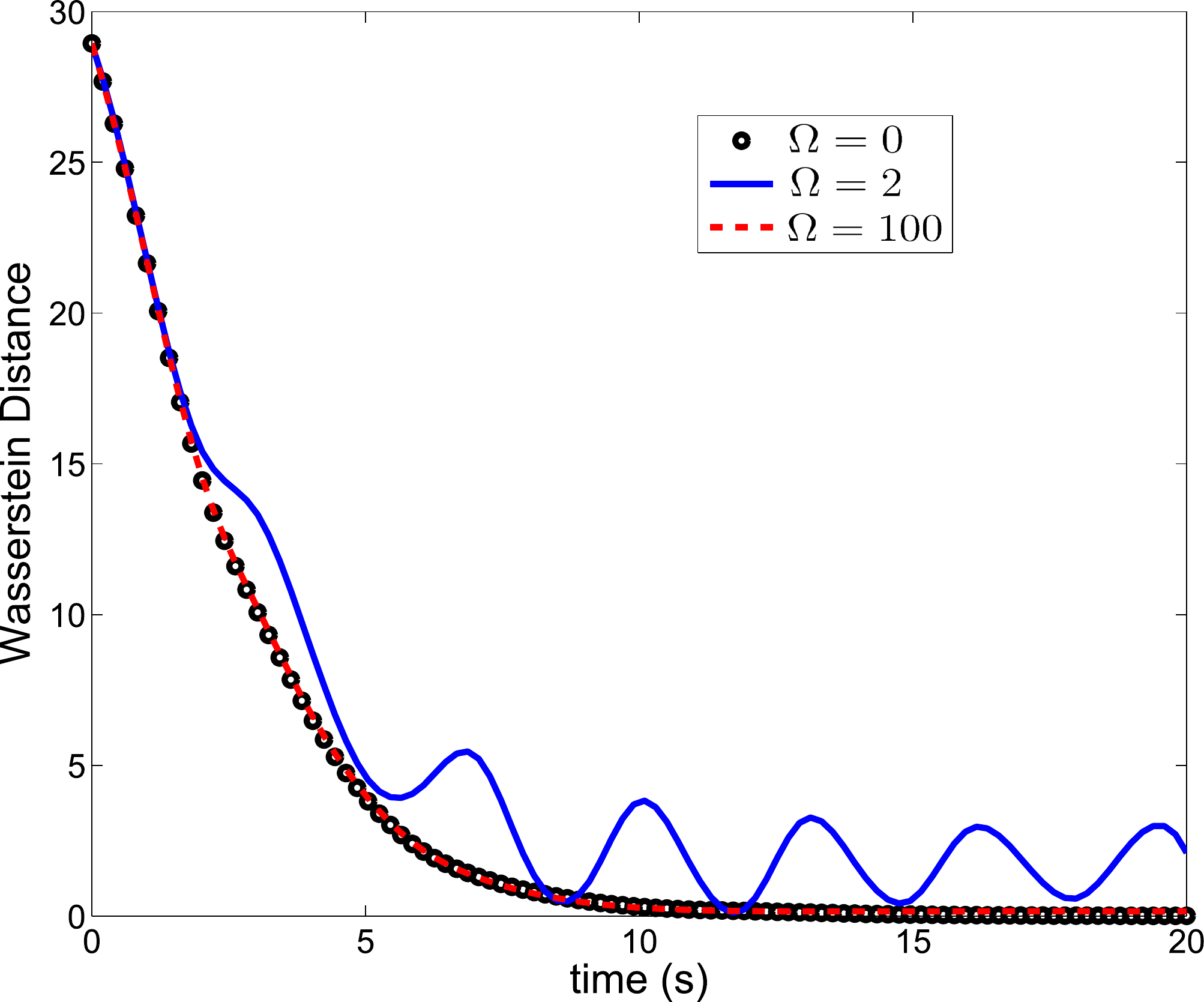}
\caption{Time evolution of Wasserstein distance for LQR, with elevator disturbance $w(t) = 6.5\sin\left(\Omega t\right)$.}
\label{LQRWwithActuatorNoise}
\end{minipage}
\hspace{0.2cm}
\begin{minipage}[b]{0.5\linewidth}
\centering
\includegraphics[width=3.3in]{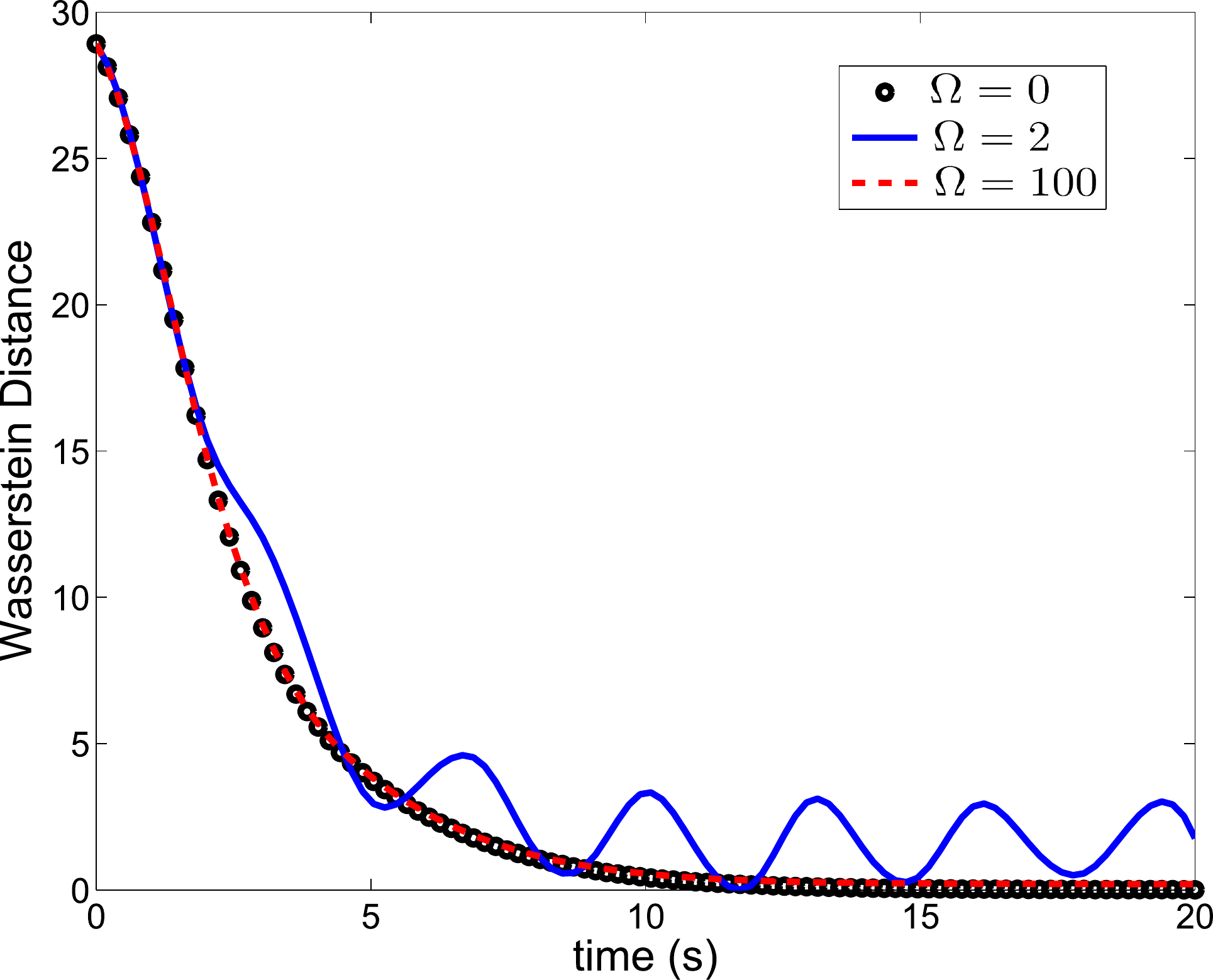}
\caption{Time evolution of Wasserstein distance for gsLQR, with elevator disturbance $w(t) = 6.5\sin\left(\Omega t\right)$.}
\label{gsLQRWwithActuatorNoise}
\end{minipage}
\end{figure}

\begin{figure}[tb]
\begin{center}
\centering
\includegraphics[scale=0.5]{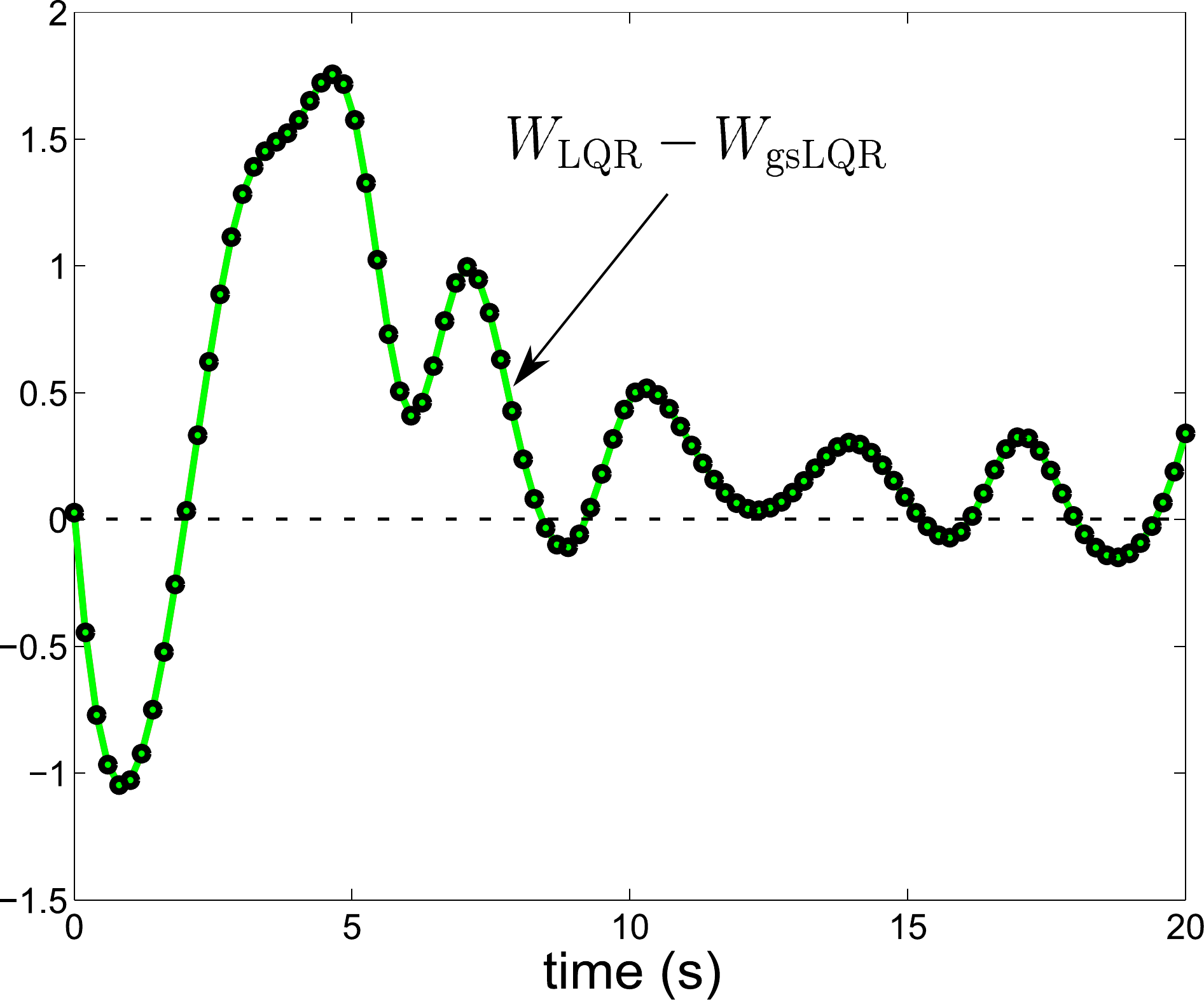}
\end{center}
\vspace*{-0.2in}
\caption{Time history of the difference between $W_{\text{LQR}}$ and $W_{\text{gsLQR}}$, with initial condition uncertainties and elevator disturbance with $\Omega = 2$ rad/s.}
\label{WassDiffActNoise}
\end{figure}

For the LQR closed-loop system, Fig. \ref{LQRWwithActuatorNoise} compares the Wasserstein distances for no actuator disturbance, i.e. $\Omega = 0$ rad/s (\emph{circles}), actuator disturbances with $\Omega = 2$ rad/s (\emph{solid line}) and $\Omega = 100$ rad/s (\emph{dashed line}), respectively. It can be seen that the Wasserstein curves for $\Omega = 0$ rad/s and $\Omega = 100$ rad/s are indistinguishable, meaning the LQR closed-loop \emph{nonlinear} system rejects high frequency elevator disturbance, similar to the \emph{linearized} closed-loop system, as observed in Fig. \ref{SingularValueFreqResponse}. For $\Omega = 2$ rad/s, the Wasserstein curve reflects the effect of closed-loop nonlinearity in joint PDF evolution till approximately $t = 10$ s. For $t > 10$ s, we observe that the LQR Wasserstein curve itself settles to an oscillation with $\omega = 2$ rad/s. This is due to the fact that by $t=10$ s, the joint probability mass comes so close to $x_{\text{trim}}$, that the linearization about $x_{\text{trim}}$ becomes a valid approximation of the closed-loop nonlinear dynamics. This observation is consistent with the MC simulations in Fig. \ref{lqr_ActNoise2}.

For the gsLQR closed-loop system, Fig. \ref{gsLQRWwithActuatorNoise} compares the Wasserstein distances with $\Omega = 0, 2, 100$ rad/s. It is interesting to observe that, similar to the LQR case, gsLQR closed loop system rejects the high frequency elevator disturbance, and hence the Wasserstein curves for $\Omega = 0$ rad/s and $\Omega = 100$ rad/s look almost identical. Further, beyond $t = 10$ s, the gsLQR closed-loop response is similar to the LQR case, and hence the respective Wasserstein curves have similar trends. However, if we compare the LQR and gsLQR Wasserstein curves for $\Omega = 2$ rad/s, then we observe that gsLQR transient performance is slightly more robust than LQR, resulting lower values of Wasserstein distance for approximately $3-5$ seconds. This transient performance difference between LQR and gsLQR, can also be seen in Fig. \ref{WassDiffActNoise} that shows the time evolution of $W_{\text{LQR}} - W_{\text{gsLQR}}$.

\section{Conclusion}
We have introduced a probabilistic framework for controller robustness verification, in the presence of stochastic initial condition and parametric uncertainties. The methodology is demonstrated on F-16 aircraft's closed-loop regulation performance with respect to two controllers: linear quadratic regulator (LQR) and gain-scheduled linear quadratic regulator (gsLQR). Compared to the current state-of-the-art, the distinguishing feature of the proposed method is that the formulation is done at the ensemble level. Hence, the spatio-temporal evolution of the joint PDF values are directly computed in exact arithmetic by solving the Liouville PDE via method-of-characteristics. Next, robustness is measured as the optimal transport theoretic Wasserstein distance between the instantaneous joint PDF and the Dirac PDF at $x_{\text{trim}}$, corresponding to the desired regulation performance. Our numerical results based on optimal transport, show that both LQR and gsLQR achieve asymptotic regulation, but the gsLQR has better transient performance. This holds for initial condition and parametric uncertainties, with or without actuator disturbance. These conclusions conform with the Monte Carlo simulations.

\appendices

\section{}
The purpose of this appendix is to prove that convergence of joint PDFs imply convergence in respective univariate marginals, but the converse is not true. Here, the convergence of PDFs is measured in Wasserstein metric. We first prove the following preparatory lemma that leads to our main result in Theorem 1.
\begin{lemma}
Let $\varphi_{1}^{i}$ and $\varphi_{2}^{i}$ be the respective $i$\textsuperscript{th} univariate marginals for $d$-dimensional joint PDFs $\varphi_{1}$ and $\varphi_{2}$, supported on $\mathbb{R}_{x_{1}} \times \mathbb{R}_{x_{2}} \times \hdots \times \mathbb{R}_{x_{d}}$, and $\mathbb{R}_{y_{1}} \times \mathbb{R}_{y_{2}} \times \hdots \times \mathbb{R}_{y_{d}}$. Let $W_{i} \triangleq W\left(\varphi_{1}^{i},\varphi_{2}^{i}\right)$, $i=1,\hdots,d$, and $\overline{W} \triangleq W\left(\varphi_{1},\varphi_{2}\right)$; then
\begin{eqnarray}
\displaystyle\sum_{i=1}^{d}W_{i}^{2} \:\leqslant\: \overline{W}^{2}.
\label{MarginalWbound}
\end{eqnarray}
\label{BoundMarginalWassSOS}
\end{lemma}
\begin{proof}
Notice that $\text{supp}\left(\varphi_{1}^{i}\right) = \mathbb{R}_{x_{i}}$, and $\text{supp}\left(\varphi_{2}^{i}\right) = \mathbb{R}_{y_{i}}$, $\forall\,i=1,\hdots,d$. For $d$-dimensional vectors $x=\left(x_{1},\hdots,x_{d}\right)^{\top}$, $y=\left(y_{1},\hdots,y_{d}\right)^{\top}$, by definition
\begin{eqnarray}
\overline{W}^{2} &=& \underset{\xi \in \mathcal{M}\left(\varphi_{1},\varphi_{2}\right)}{\text{inf}} \displaystyle\int_{\mathbb{R}^{2d}}\parallel x - y \parallel_{2}^{2} \: \xi\left(x,y\right) \: dx \: dy
= \displaystyle\int_{\mathbb{R}^{2d}}\parallel x - y \parallel_{2}^{2} \: \xi^{\star}\left(x,y\right) \: dx \: dy,
\end{eqnarray}
where $\xi^{\star}\left(x,y\right)$ is the optimal transport PDF supported on $\mathbb{R}^{2d}$. Clearly,
\begin{eqnarray}
\varphi_{1}^{i} &=& \displaystyle\int_{\mathbb{R}^{2d-1}} \xi^{\star}\left(x,y\right) dx_{1} \hdots dx_{i-1} dx_{i+1} \hdots dx_{d} dy_{1} \hdots dy_{d}, \\
\varphi_{2}^{i} &=& \displaystyle\int_{\mathbb{R}^{2d-1}} \xi^{\star}\left(x,y\right) dx_{1} \hdots dx_{d} dy_{1} \hdots dy_{i-1} dy_{i+1} \hdots dy_{d}.
\end{eqnarray}
Thus, we have
\begin{eqnarray}
W_{i}^{2} &=& \underset{\eta \in \mathcal{M}\left(\varphi_{1}^{i},\varphi_{2}^{i}\right)}{\text{inf}} \displaystyle\int_{\mathbb{R}^{2}} \left(x_{i} - y_{i}\right)^{2} \: \eta\left(x_{i},y_{i}\right) \: dx_{i} \: dy_{i}, \nonumber\\
&=& \displaystyle\int_{\mathbb{R}^{2}} \left(x_{i} - y_{i}\right)^{2} \: \eta^{\star}\left(x_{i},y_{i}\right) \: dx_{i} \: dy_{i}, \nonumber\\
&\leqslant& \displaystyle\int_{\mathbb{R}^{2}} \left(x_{i} - y_{i}\right)^{2} \: \widetilde{\xi}^{\star}\left(x_{i},y_{i}\right) \: dx_{i} \: dy_{i},
\label{IntmInequality}
\end{eqnarray}
where $\widetilde{\xi}^{\star}\left(x_{i},y_{i}\right)$ is the $\left(i,i\right)$\textsuperscript{th} bivariate marginal of $\xi^{\star}\left(x,y\right)$. Since $\displaystyle\sum_{i=1}^{d} \left(x_{i} - y_{i}\right)^{2} = \parallel x - y \parallel_{2}^{2}$, the result follows from (\ref{IntmInequality}), after substituting
\begin{eqnarray}
\widetilde{\xi}^{\star}\left(x_{i},y_{i}\right) = \displaystyle\int_{\mathbb{R}^{2d-2}} \xi^{\star}\left(x,y\right) dx_{1} \hdots dx_{i-1} dx_{i+1} \hdots dx_{d}\:dy_{1} dy_{i-1} dy_{i+1} \hdots dy_{d}.
\end{eqnarray}
This completes the proof.
\end{proof}
\begin{theorem}
Convergence of Joint PDFs in Wasserstein metric, implies convergence of univariate marginals. Converse is not true.
\end{theorem}
\begin{proof}
Using the notation of Lemma \ref{BoundMarginalWassSOS}, when the joints $\varphi_{1}$ and $\varphi_{2}$ converge, then $\overline{W} = 0$. Hence from (\ref{MarginalWbound}), $\displaystyle\sum_{i=1}^{d} W_{i}^{2} = 0 \Rightarrow W_{i} = 0$, $\forall\:i = 1,\hdots,d$. However, $W_{i} = 0 \Rightarrow \overline{W} \geqslant 0$. Hence the result.
\end{proof}
\begin{remark}
In our context of comparing LQR and gsLQR robustness, (\ref{MarginalWbound}) yields
\begin{eqnarray}
\left(\sum_{j=1}^{n_{x}=4} W_{j}^{2}\left(\varphi_{\text{LQR}}^{j},\varphi_{\text{gsLQR}}^{j}\right)\right)^{1/2} \leqslant W\left(\varphi_{\text{LQR}},\varphi_{\text{gsLQR}}\right) \leqslant W\left(\varphi_{\text{LQR}},\varphi^{*}\right) + W\left(\varphi_{\text{gsLQR}},\varphi^{*}\right),
\label{LQRgsLQRconvergence}
\end{eqnarray}
where the last step is due to triangle inequality. From Fig. \ref{wdist}, we observe that with time, both $W\left(\varphi_{\text{LQR}},\varphi^{*}\right)$ and $W\left(\varphi_{\text{gsLQR}},\varphi^{*}\right)$ converge to zero. As a consequence, $W_{j}\left(\varphi_{\text{LQR}}^{j},\varphi_{\text{gsLQR}}^{j}\right) = 0$ (from (\ref{LQRgsLQRconvergence})), $j=1,\hdots,4$, as evidenced by Fig. \ref{1dmarg}.

\end{remark}

\ifCLASSOPTIONcaptionsoff
  \newpage
\fi

\end{document}